\documentclass[12pt]{article}

\usepackage{times}

\usepackage{xcolor}
\usepackage{graphicx}
\usepackage{epstopdf}
\usepackage{url}
\usepackage{amssymb}
\usepackage{amsmath}
\usepackage{amsfonts}
\usepackage{amsthm}
\usepackage[normalem]{ulem}
\usepackage{stmaryrd}

\newtheorem{Theorem}{Theorem}
\newtheorem{Lemma}{Lemma}

\usepackage[displaymath, mathlines]{lineno}

\usepackage{cite}

\usepackage{setspace}

\renewcommand{\eqref}[1]{Eq.~\ref{#1}}

\newif\ifcmnt
\cmnttrue
\ifdefined\cmntsoff\cmntfalse\fi

\ifcmnt
    \providecommand{\aucmnt}[1]{#1}
\else
    \providecommand{\aucmnt}[1]{}
\fi

\newcommand{\ignore}[1]{}

\newcommand{\bfABXY}{\mathbf{ABXY}}
\newcommand{\Vthresh}{v_{\mathrm{thresh}}}

\title{Experimentally Generated Randomness Certified by the Impossibility of Superluminal Signals}

\author
{Peter Bierhorst,$^{1\ast}$ Emanuel Knill,$^{1,6}$  Scott Glancy,$^{1}$ Yanbao Zhang,$^{1\dagger}$\\
 Alan Mink,$^{2,3}$ Stephen Jordan,$^{2}$ Andrea Rommal,$^{4}$ Yi-Kai Liu,$^{2}$ \\
 Bradley Christensen,$^{5}$ Sae Woo Nam,$^{1}$ Martin J. Stevens,$^{1}$  Lynden K. Shalm$^{1}$\\
\\
\normalsize{$^{1}$National Institute of Standards and Technology,}
\normalsize{Boulder 80305, CO, USA}\\
\normalsize{$^{2}$ National Institute of Standards and Technology,}
\normalsize{Gaithersburg 20899, MD, USA}\\
\normalsize{$^{3}$ Theiss Research,}
\normalsize{La Jolla, CA, 92037, USA}\\
\normalsize{$^{4}$ Muhlenberg College, Allentown, PA, 18104, USA}\\
\normalsize{$^{5}$ Department of Physics, University of Wisconsin, Madison, WI,
53706, USA}\\
\normalsize{$^{6}$Center for Theory of Quantum Matter, University of Colorado, Boulder, Colorado 80309, USA}\\
\normalsize{$^{\dagger}$Present address: NTT Basic Research Laboratories and NTT Research Center for Theoretical}\\ 
\normalsize{Quantum Physics, NTT Corporation, 3-1 Morinosato-Wakamiya, Atsugi, Kanagawa 243-0198, Japan}\\
\normalsize{$^\ast$ E-mail:  peter.bierhorst@nist.gov}
}

\date{}

\begin{document}



\maketitle


{\bf From dice to modern complex circuits, there have been many
  attempts to build increasingly better devices to generate random
  numbers. Today, randomness is fundamental to security and
  cryptographic systems, as well as safeguarding privacy. A key
  challenge with random number generators is that it is hard to ensure
  that their outputs are
  unpredictable\cite{acin:2016,pironio:2013,miller:2014}. For a random
  number generator based on a physical process, such as a noisy
  classical system or an elementary quantum measurement, a detailed
  model describing the underlying physics is required to assert
  unpredictability. Such a model must make a number of assumptions
  that may not be valid, thereby compromising the integrity of the
  device. However, it is possible to exploit the phenomenon of
  quantum nonlocality with a loophole-free Bell test to build a random
  number generator that can produce output that is unpredictable to
  any adversary limited only by general physical
  principles \cite{acin:2016,pironio:2013, miller:2014,colbeck:2011,
    pironio:2010, vazirani:2012, fehr:2013, chung:2014,
    nieto:2014,bancal:2014, thinh:2016}. With recent technological
  developments, it is now possible to carry out such a loophole-free
  Bell test \cite{hensen:2015,shalm:2015,giustina:2015}. Here we
  present certified randomness obtained from a photonic Bell
  experiment and extract 1024 random bits uniform to within
  $10^{-12}$. These random bits could not have
    been predicted within any physical theory that prohibits
  superluminal signaling and allows one to make independent
  measurement choices. To certify and quantify the randomness, we
  describe a new protocol that is optimized for apparatuses
  characterized by a low per-trial violation of Bell
  inequalities. We thus enlisted an experimental
  result that fundamentally challenges the notion of determinism to
  build a system that can increase trust in random sources. In the
  future, random number generators based on loophole-free Bell tests
  may play a role in increasing the security and trust of our
  cryptographic systems and infrastructure.}


 The search for certifiably unpredictable random number generators is
motivated by applications, such as secure communication, for which the
predictability of pseudorandom strings make them unsuitable. Private randomness is required to initiate and authenticate virtually
every secure communication \cite{paar2010}, and public randomness from
randomness beacons can be used for public certification and resource
distribution in many settings \cite{fischer:2011}. To certify
randomness, one can perform an experiment known as a Bell
test \cite{BELL}, which in its simplest form performs measurements on
an entangled system located in two physically separated measurement
stations, with each station choosing between two types of
measurements. After multiple experimental trials with varying
measurement choices, if the measurement data violates conditions known
as ``Bell inequalities,'' then the data can be certified to contain
randomness under weak assumptions. 

Our randomness generation employs a
``loophole-free'' Bell test, which notably is characterized by high
detection efficiency and space-like separation of the measurement
stations during each experimental trial.  The bits are unpredictable
assuming that (1) the choices of measurement settings are independent
of the experimental devices and pre-existing classical information
about them and (2) in each experimental trial, the measurement
outcomes at each station are independent of the settings choices at
the other station. The first assumption is ultimately untestable, but
the premise that it is possible to choose measurement settings
independently of a system being measured is often tacitly invoked in
the interpretation of many scientific experiments and laws of physics
\cite{bell:1985}.  The second assumption can only be violated if one
admits a theory that permits sending signals faster than the speed of
light, given our trust that the space-like separation of the relevant
events in the experiment is accurately verified by the timing
electronics and that results are final when recorded. We also trust
that the classical computing equipment used to process the data
operates according to specification.  

Under the above assumptions, the
output randomness is certified to be unpredictable with respect to a
real or hypothetical actor ``Eve'' in possession of the pre-existing
classical information, physically isolated from the devices while they
are under our control, and without access to data produced during
  the protocol. The bits remain unpredictable to Eve if she learns
the settings at any time after her last interaction with the devices.
If the devices are trusted, which is reasonable if we built them, then
this may be well before the start of the protocol, in which case the
settings can come from public randomness \cite{pironio:2013,bancal:2014}. In particular, one can use
  an existing public randomness source, such as the NIST random beacon \cite{fischer:2011}, to generate much needed private randomness as output. Since the assumptions do not constrain the
specific physical realization of the devices and do not require
specific states or measurements, they implement a
``device-independent'' framework \cite{colbeck:2007} which allows
  an individual user to assure security with minimal assumptions about
  the devices. If Eve has quantum memory, it
  is possible to ensure that Eve's side information is effectively
  classical by verifying that the devices have no long-term quantum
  memory of past interactions with Eve. While this introduces weak
  device-dependence, for the foreseeable future this verification task
  is comparable to that required to enforce the absence of
  communication from the devices to Eve.

The only previous experimental production of certified randomness from
Bell test data was reported in the ground-breaking paper by Pironio et
al. \cite{pironio:2010}.  Their Bell test was implemented with ions in
two separate ion-traps, closing the detection loophole
\cite{pearle:1970} but without space-like separation. Indeed, Bell
tests achieving space-like separation without other experimental
loopholes have been performed only recently \cite{hensen:2015,
  shalm:2015, giustina:2015,rosenfeld:2017}. Under more restrictive
assumptions than ours, the maximum amount of randomness in principle
available in the data of Pironio et al. was quantified as $42$ bits
with an error parameter of $0.01$, but they did not extract a
uniformly distributed bit string from their data. Pironio et
al. argue that any interaction between measurement stations in their
experiment is negligible, because they are located in separate
ion-traps, each in its own vacuum chamber. However, any shielding
between the stations is necessarily incomplete; for example they must
have an open quantum channel to establish entanglement. Mundane
physical effects can allow local-realistic systems to appear to
violate Bell inequalities when shielding is incomplete.  Relying
instead on the impossibility of faster-than-light communication
provides stronger assurance of the unpredictability of the randomness.
  
We generated randomness using an improved version of the
  loophole-free Bell test reported in Ref.~[13]. Five new
  data sets were collected, with the best-performing data set yielding
  1024 new random bits uniform to within $10^{-12}$. We also
    obtained 256 random bits from the main data set analyzed in
    Ref.~[13], albeit only uniform to within $0.02$. The experiment, illustrated in
  Fig. \ref{expschematic}, consisted of a source of entangled
photons and two measurement stations named ``Alice'' and ``Bob''.
During an experimental trial, at each station a random choice was made
between two measurement settings labeled 0 and 1, after which a
measurement outcome of detection (+) or nondetection (0) was
recorded. Each station's implementation of the measurement setting was
space-like separated from the other station's measurement event, and
no postselection was employed in collecting the data. See the Methods
section for details.  For trial $i$, we model Alice's settings choices
with the random variable $X_i$ and Bob's with $Y_i$, both of which
take values in the set $\{0,1\}$.  Alice's and Bob's measurement
outcome random variables are respectively $A_i$ and $B_i$, both of
which take values in the set $\{\text{+},0\}$. When referring to a
generic single trial, we omit indices. With this notation, a general
Bell inequality for our scenario can be expressed in the form
\cite{BBP}
\begin{equation}\label{e:genbellineq}
\sum_{abxy}s^{ab}_{xy}\mathbb{P}(A=a,B=b|X=x,Y=y) \le \beta,
\end{equation}
where the $s^{ab}_{xy}$ are fixed real coefficients indexed by
$a,b,x,y$ that range over all possible values of $A,B,X,Y$.  The upper
bound $\beta$ is required to be satisfied whenever the
settings-conditional outcome probabilities are induced by a model
satisfying ``local realism'' (LR). LR distributions, which cannot be
certified to contain randomness, are those for which
$\mathbb{P}(A=a,B=b|X=x,Y=y)$ is of
the form $\sum_\lambda \mathbb{P}(A=a|X=x,
\Lambda=\lambda)\mathbb{P}(B=b|Y=y,
\Lambda=\lambda)\mathbb{P}(\Lambda=\lambda)$ for a random variable
$\Lambda$ representing local hidden variables. The Bell inequality is
non-trivial if there exists a quantum-realizable distribution that can
violate the bound $\beta$.

It has long been known that experimental violations of Bell inequalities such as \eqref{e:genbellineq} indicate the presence
of randomness in the data. To quantify randomness
with respect to Eve, we represent Eve's initial classical information
by a random variable $E$.  We formalize the assumption that
measurement settings can be generated independently of the system
being measured and Eve's information with the following condition:
\begin{equation}\label{e:mtunifsettings}
\mathbb{P}(X_i=x,Y_i=y|E=e,\text{past}_i)=\mathbb{P}(X_i=x,Y_i=y)=\frac{1}{4} \quad \forall x,y,e,
\end{equation}
where $\text{past}_i$ represents events in the past of the $i$'th
trial, specifically including the trial settings and outcomes for
trial 1 through $i-1$. Our other assumption, that measurement outcomes
are independent of remote measurement choices, is formalized as
follows:
\begin{eqnarray}\label{e:mtnosig}
      \mathbb{P}(A_i=a|X_i=x,Y_i=y,E=e, \text{past}_i) &=& \mathbb{P}(A_i=a|X_i=x,E=e, \text{past}_i)\notag\\
      \mathbb{P}(B_i=b|X_i=x,Y_i=y,E=e, \text{past}_i) &=& \mathbb{P}(B_i=b|Y_i=y,E=e, \text{past}_i)\quad \forall x,y,e.
\end{eqnarray}
These equations are commonly referred to as the ``non-signaling''
assumptions, although they are often stated without the conditionals
$E$ and $\text{past}_i$.  Our space-like separation of settings and
remote measurements provide assurance that the experiment obeys
Eqs.~\ref{e:mtnosig}. We remark that if one assumes the measured systems obey
  quantum physics, stronger constraints are possible \cite{TBOUND,navascues:2008}.

Given Eqs.~\ref{e:mtunifsettings} and \ref{e:mtnosig}, our protocol
produces random bits in two sequential parts. For the first part,
``entropy production'', we implement $n$ trials of the Bell test, from
which we compute a statistic $V$ related to a Bell inequality
(\eqref{e:genbellineq}). $V$ quantifies the Bell violation and
determines whether or not the protocol passes or aborts.  If the
protocol passes, we certify an amount of randomness in the outcome
string even conditioned on the setting string and $E$. In the second part,
``extraction,'' we process the outcome string into a shorter string of
bits whose distribution is close to uniform. We used our customized
implementation of the Trevisan extractor \cite{trevisan:2001} derived
from the framework of Mauerer, Portmann and Scholz \cite{mauerer:2012}
and the associated open source code. We call this the TMPS algorithm,
see Supplementary Information (SI) \ref{st:trevisan} for details.

We applied a new method of certifying the amount of
  randomness in Bell tests.  Previous methods for related models with various sets of assumptions \cite{pironio:2013,miller:2014,colbeck:2011,pironio:2010,vazirani:2012,fehr:2013,chung:2014,coudron:2014,dupuis:2016,arnon:2016}
are ineffective in our experimental regime (SI \ref{st:previous}),
which is characterized by a small per-trial violation of Bell
inequalities. Other recent
  works that explore how to effectively certify randomness from a
wider range of experimental regimes assume that measured states are
independent and identically distributed (i.i.d.)  or that the regime
is asymptotic \cite{nieto:2014,bancal:2014,thinh:2016, miller2:2014}.
Our method, which does not require these assumptions, builds on the
Prediction-Based Ratio (PBR) method for rejecting LR
\cite{zhang:2011}.  Applying this method to training data (see below),
we obtain a real-valued ``Bell function'' $T$ with arguments $A,B,X,Y$
that satisfies $T(A,B,X,Y) > 0$ with expectation $\mathbb{E}(T) \leq
1$ for any LR distribution satisfying Eq.~\ref{e:mtunifsettings}. From
$T$ we determine the maximum value $1+m$ of $\mathbb{E}(T)$ over all
distributions satisfying Eqs.~\ref{e:mtunifsettings}
and~\ref{e:mtnosig}, where we require that $m>0$.  Such a function $T$
induces a Bell inequality (\eqref{e:genbellineq}) with $\beta=4$ and
$s^{ab}_{xy}= T(a,b,x,y)$.  Define $T_i=T(A_i,B_i,X_i,Y_i)$ and
$V=\prod_{i=1}^nT_i $. If the experimenter observes a value of $V$
larger than $1$, this indicates a violation of the Bell inequality and
the presence of randomness in the data.  The randomness is quantified
by the following theorem, proven in the SI \ref{st:ept}. Below, we
denote all of the settings of both stations with ${\bf XY} =
X_1Y_1X_2Y_2...X_nY_n$, and other sequences such as $\mathbf{AB}$ and
$\mathbf{ABXY}$ are similarly interleaved over $n$ trials.

\medskip
\medskip

\noindent{\it Entropy Production Theorem.}  Suppose $T$ is a Bell function
satisfying the above conditions. Then in an experiment of $n$ trials
obeying Eqs.~\ref{e:mtunifsettings} and~\ref{e:mtnosig}, the following
inequality holds for all $\epsilon_{\text{p}} \in (0,1)$ and
$\Vthresh$ satisfying $1\le \Vthresh \le (1+(3/2)m)^{n}\epsilon_{\mathrm{p}}^{-1}$:
\begin{equation}\label{e:theorem}
\mathbb{P}_e\left(\mathbb{P}_e({\bf AB}|{\bf XY})> \delta \text{ AND } V\ge \Vthresh \right)\le \epsilon_{p}
\end{equation}
where $\delta = [1+(1-\sqrt[n]{\epsilon_{\text{p}}\Vthresh})/(2m)]^n$
and $\mathbb{P}_e$ denotes the probability distribution conditioned on the
event $\{E=e\}$, where $e$ is arbitrary. The expression
$\mathbb{P}_e({\bf AB}|{\bf XY})$ denotes the random variable that takes the
value $\mathbb{P}_{e}(\mathbf{AB}=\mathbf{ab}|\mathbf{XY}=\mathbf{xy})$ when
$\mathbf{ABXY}$ takes the value $\mathbf{abxy}$.

\medskip
\medskip

\noindent In words, the theorem says that with high probability, if
$V$ is at least as large as $\Vthresh$, then the output $\mathbf{AB}$
is unpredictable, in the sense that no individual outcome
$\{\mathbf{AB}=\mathbf{ab}\}$ occurs with probability higher than
$\delta$, even given the information
$\{\mathbf{XY}E=\mathbf{xy}e\}$. The theorem supports a protocol that
aborts if $V$ takes a value less
than $\Vthresh$, and passes otherwise. If the probability of passing
were 1, then $-\log_2(\delta)$ would be a so-called ``smooth
min-entropy'', a quantity that characterizes the number of uniform
bits of randomness that are in principle available in $\mathbf{AB}$
\cite{trevisan:2000, renner:2006}. We show in the SI~\ref{st:T} that
for constant $\epsilon_{\text{p}}$, $-\log_{2}(\delta)$ is
proportional to the number of trials. How many bits we can actually
extract depends on $\epsilon_{\text{fin}}$, the final output's maximum
allowed distance from uniform. We also show in the SI that the Entropy
Production Theorem can still be proved if \eqref{e:mtunifsettings} is
weakened so that settings probabilities need not be known but are
constrained to be within $\alpha$ of $1/4$ with $\alpha<1/4$, while
still being conditionally independent of earlier outcomes given
earlier settings. Such a weakening is relevant for experiments
\cite{hensen:2015,shalm:2015,giustina:2015} that use physical random
number generators to choose the settings, for which the settings
probabilities cannot be known exactly.

To extract the available randomness in ${\bf AB}$, we use the TMPS
algorithm to obtain an extractor, specifically a function $\text{Ext}$
that takes as input the string ${\bf AB}$ and a length $d$ ``seed''
bit string ${\bf S}$, where ${\bf S}$ is uniform and independent of
${\bf ABXY}$. Its output is a length $t$ bit string. ${\bf
  S}$ can be obtained from $d$ additional instances of the random
variables $X_i$, so \eqref{e:mtunifsettings} ensures the needed
independence and uniformity conditions on ${\bf S}$.  In order for the
output to be within a distance $\epsilon_{\text{fin}}$ of uniform
independent of ${\bf XY}$ and $E$, the entropy production and
extractor parameters must satisfy the constraints given in the next
theorem, proven in the SI \ref{st:pst}. In the statement of
the theorem, the measure of distance used
is the ``total variation (TV) distance,'' expressed by the left side of \eqref{e:mtfinal}, and ``$\mathrm{pass}$'' is the event that $V$ exceeds $\Vthresh$.

\medskip\medskip

\noindent{\it Protocol Soundness Theorem.}
Let $0<\epsilon_{\text{ext}}, \kappa<1$. Suppose that $\mathbb P(\text{pass})\ge\kappa$ and suppose that that the protocol parameters satisfy 
\begin{equation}\label{e:mttrev1}
t+4\log_2t \le -\log_2 \delta + \log_2 \kappa +5\log_2 \epsilon_{\text{ext}} -11.
\end{equation}
Then the output ${\bf U} = \text{Ext}({\bf
    AB},{\bf S})$ of the function obtained by the TMPS algorithm satisfies
\begin{multline}\label{e:mtfinal}
  \frac{1}{2}\sum_{{\bf u},{\bf xys}e} \Big|\mathbb{P}\big({\bf U}={\bf u},{\bf XYS}E={\bf
    xys}e|\text{pass}\big)-\mathbb{P}^{\text{unif}}({\bf U}={\bf u})\mathbb{P}\big({\bf
    XY}E={\bf xy}e|\text{pass}\big)\mathbb{P}^{\text{unif}}({\bf S}={\bf s})\Big|\\ \le
  \epsilon_{\mathrm{p}}/\mathbb P(\mathrm{pass})+\epsilon_{\mathrm{ext}},
\end{multline}
  where $\mathbb{P}^{\mathrm{unif}}$ denotes the uniform probability
  distribution.

\medskip
\medskip

\noindent
The number of seed bits $d$ required satisfies $d
=O(\log(t)\log(nt/\epsilon_{\mathrm{ext}})^{2})$, and SI
  \ref{st:trevisan} gives an explicit bound.

The theorem provides several options for quantifying the
  uniformity of the randomness produced. A goal is for the protocol to
  be nearly indistinguishable according to TV distance from an ideal
  protocol, where in an ideal protocol the randomness is perfectly
  uniform conditional on passing. For this, the ideal protocol can be
  chosen to have the same probability of passing with behavior matching that of the real protocol when aborting. The theorem implies that the unconditional distribution of the protocol is within TV distance
  $\max(\epsilon_{\mathrm{p}}+\epsilon_{\mathrm{ext}},\kappa)$ of that of an
  ideal protocol (SI \ref{st:pst}).  For this distance, if the probability of passing is
  comparable to $\kappa$, then the conditional TV distance from
  uniform, given in \eqref{e:mtfinal}, could be large. It is desirable that even for the worst case
  probability of passing, the conditional TV distance be small.
  Accordingly, we quantify the uniformity for our implementation with
  $\epsilon_{\mathrm{fin}}=\max(\epsilon_{\mathrm{p}}/\kappa +
  \epsilon_{\mathrm{ext}},\kappa)$.  Then, for any probability of
  passing greater than $\epsilon_{\mathrm{fin}}$, conditionally on
  passing, the TV distance from uniform is at most
  $\epsilon_{\mathrm{fin}}$. 

We applied our protocol to five data sets using the setup based on that described in Ref.~[13]
with improvements described in the Methods section. Each data set was
collected in five to ten minutes, improving on the approximately
one month duration of data acquisition reported in Ref.~[5]. Before starting
the protocol, we set aside the first $5\times 10^{6}$ trials of each
data set as training data, which we used to choose parameters needed
by the protocol.  With the training data removed, the number $n$ of
trials used by the protocol was between $2.5\times 10^7$ and $5.5 \times 10^7$
  for each data set. We used the training data to determine a Bell
function $T$ with statistically strong violation of LR on the training
data according to the PBR method \cite{zhang:2011}; see SI
\ref{st:T}. The function $T$ obtained for the fifth data set, which was longest in duration and produced the most randomness, is given in Table
\ref{t:PBR} as an example. We computed thresholds $\Vthresh$ so that a
sample of $n$ i.i.d. trials from the distribution inferred from the
training data would have a high probability for exceeding $\Vthresh$.

\begin{table}
  \centering\caption{{\bf Bell function $T$ obtained from Data Set 5.} We used a numerical method based on maximum likelihood
    to infer a non-signaling
    distribution based on the raw counts of the training trials,
    namely the first $5\times 10^{6}$ trials. We then
    determined the function $T$ that maximizes $\mathbb{E}(\ln T)$
    according to this distribution, subject to the constraints that
    $\mathbb{E}(T)_{LR}\le 1$ for all LR distributions and
    $T(0,0,x,y)=1$ for all $x,y$. The latter constraint
      improves the signal-to-noise for our data. The function $T$
    yields $m=0.0100425$, and $\mathbb{E}(T) = 1.000003931$ for the
    non-signaling distribution inferred from the training data. One
    can also interpret the numbers below as the coefficients
    $s^{ab}_{xy}$ in \eqref{e:genbellineq}, which defines a Bell
    inequality with $\beta=4$. The values of $T$ are rounded down at
    the tenth digit.}
\label{t:PBR}
 \begin{tabular}{ r|c|c|c|c| }
 \multicolumn{1}{r}{}
  &  \multicolumn{1}{c}{$ab=\text{++}$}
 &  \multicolumn{1}{c}{$ab=\text{+}0$}
 &  \multicolumn{1}{c}{$ab=0\text{+}$} 
   &  \multicolumn{1}{c}{$ab=00$}
\\
  \cline{2-5}
   $xy=00$&   1.0243556353  &  0.9704647804  &  0.9735507658  & 1 \\
 \cline{2-5}
   $xy=01$&   1.0256127409  &  0.9491951243  &  0.9960775334  & 1 \\
 \cline{2-5}
   $xy=10$&   1.0227274988  &  0.9962782754  &  0.9461091383  & 1 \\
 \cline{2-5}
   $xy=11$&   0.9273040563  &  1.0037217225  &  1.0039224645  & 1 \\
 \cline{2-5}
 \end{tabular}
\end{table} 

For the fifth data set, a sample of $n$ i.i.d.\ trials from the
distribution inferred from the training data would have approximately
0.99 probability of exceeding a threshold of $\Vthresh
= 1.5\times 10^{32}$. This would allow the extraction of 1024 bits
uniform to within $\epsilon_{\text{fin}}=10^{-12}$, using
$\epsilon_{\mathrm{p}} =\kappa^2=9.025\times 10^{-25}$ and
$\epsilon_{\mathrm{ext}}=5 \times 10^{-14}$. These values were chosen
based on a numerical study of the constraints on the number $t$ of
bits extracted for fixed values of
$\epsilon_{\mathrm{fin}}=10^{-12}$. Running the protocol on the
remaining $55,110,210$ trials with these parameters, the product
$\prod_{i=1}^nT_i$ exceeded $\Vthresh$, and so the protocol passed.
Applying the extractor to the resulting output string ${\bf AB}$ with
a seed of length $d=315,844$, we extracted $1024$ bits, certified to
be uniform to within $10^{-12}$, the first ten of which are
1110001001. Figure \ref{f:FEB} displays the extractable bits for
alternative choices of $\epsilon_{\text{fin}}$ for all five data sets.

We also applied the protocol to data from the experiment of
Ref.~[13].  This experiment was more conservative in
taking additional measures to ensure that it was loophole-free,
including space-like separation of the measurement choices from both
the downconversion event and the remote measurement outcomes.  We
extracted $256$ bits at $\epsilon_{\mathrm{fin}}=0.02$ from the best
data set, XOR 3, reported in Ref.~[13]. The distance from
an ideal protocol as explained after the Protocol Soundness Theorem
was $4.00\times 10^{-4}$, without accounting for possible bias in the
random source used. For details see SI \ref{st:actual}.

For the data set producing 1024 new near random bits, our protocol
used $1.10 \times 10^8$ uniform bits to choose the settings and
$3.16\times 10^5$ uniform bits to choose the seed.  Because the
extractor used is a ``strong'' extractor, the seed bits are still
uniform conditional on passing, so they can be recovered at the end of
the protocol for uses elsewhere. This is not the case for the
settings-choice bits because the probability of passing is less than
$1$. To reduce the entropy used for the settings, our protocol can be
modified to use highly biased settings choices \cite{pironio:2010}. Reducing settings entropy is not a priority if the settings and seed bits come from a public source of randomness,
  in which case
  the output bits can still  be certified to be unknown to external
  observers such as Eve 
  and the current protocol is an
  effective method for private randomness generation~\cite{pironio:2013,bancal:2014}.

For future work, we hope to take advantage of the adaptive
capabilities of the Entropy Production Theorem (SI~\ref{st:ept}) to
dynamically compensate for experimental drift during run time. In view
of advances toward practical quantum computing it is desirable to
study the protocol in the presence of quantum side information, which
may require more conservative randomness generation. We also look
forward to technical improvements in experimental equipment for larger
violation and higher trial rates. These will enable faster generation
of random bits with lower error and support the use of biased settings
choices.

Existing randomness generation systems rely on detailed assumptions
about the specific physics underlying the devices. With the advent of
loophole-free Bell tests, it is now possible to build quantum devices
that exploit quantum nonlocality to remove many of the
device-dependent assumptions in current technological
implementations. Our device-independent random number generator is an
example of such a system. Such generators can provide the best method
currently known for physically producing randomness, thereby improving
the security of a wide range of applications.

\medskip

\medskip

\noindent{\bf Methods} We used
polarization-entangled photons generated by a nonlinear crystal pumped
by a pulsed, picosecond laser at approximately 775 nm in a
configuration similar to that reported in Ref.~[13], but with
several improvements to increase the rate of randomness
extraction. The laser's repetition rate was 79.3 MHz, and each pulse
that entered the crystal had a probability of $\approx 0.003$ of
creating an entangled photon pair in the state $\left|\psi
\right\rangle \approx 0.982 \left|HH \right\rangle + 0.191 \left|VV
\right\rangle$ at a center wavelength of 1550 nm. By pumping the
crystal with approximately five times as much power, and using a 20 mm
long crystal, we were able to substantially increase the per-pulse
probability of generating a downconversion event compared with
Ref.~[13] while maintaining similar overall system
efficiencies. The two entangled photons from each pair were
separately sent to one of the two
measurement stations ($187 \pm 1$) m apart. At Alice and Bob, a Pockels cell and polarizer
combined to allow the rapid switching of measurement bases and
measurement of the polarization state of the incoming photons. Each
Pockels cell operated at a rate of 100 kHz, allowing us to perform
100,000 trials per second (the driver electronics on the Pockels cells
sets this rate). The photons were then detected using fiber-coupled
superconducting single-photon nanowire detectors, with Bob's detector
operating at approximately $90 \%$ efficiency and Alice's detector
operating with approximately $92 \%$ efficiency
\cite{marsili:2013}. For this experiment, the total symmetric system
heralding efficiency was $(75.5 \pm 0.5 \%)$, which is above the $71.5
\%$ threshold required to close the detection-loophole for our
experimental configuration after accounting for unwanted background
counts at our detectors and slight imperfections in our state
preparation and measurements components.

With this configuration, Bob completed his measurement
$(294.4 \pm 3.7)$ ns before a hypothetical switching signal travelling
at light speed from Alice's Pockels cell could arrive at his
station. Similarly, Alice completed her measurement $(424.2 \pm 3.7)$
ns before such a signal from Bob's Pockels cell could arrive at her
location. Each trial's outcome values were obtained by aggregating the
photon detection or non-detection events from several short time
intervals lasting 1024 ps, each of which is timed to correspond to one
pulse of the pump laser.  If any photons were detected in the short
intervals, the outcome is ``+'', and if no photons were detected, the
outcome is ``0''.  The experiment of Ref.~[13] used at most 7
short intervals, but here we were able to include 14 intervals while
maintaining space-like separation, which further increased the
probability of observing a photon during each trial. For demonstration
purposes, Alice and Bob each used Python's \texttt{random.py} module
with the default generator (the Mersenne twister) to pick their
settings at each trial. This pseudorandom source is predictable,
  and for secure applications of the protocol in an adversarial
  scenario, such as if the photon
    pair source or measurement devices are obtained from
    an untrusted provider, settings choices must be based on random
  sources that are effectively not predictable. However, based on our
  knowledge of device construction, we know that our devices have
    no physical resources for predicting pseudo-random numbers and
  expect that measurement settings were
  effectively independent of relevant devices so that
  Eqs.~\ref{e:mtunifsettings} and~\ref{e:mtnosig} still hold. We remark that the
settings choices for the XOR 3 data set were based on physical random
sources.

With the improved detection efficiency, the higher per-trial
probability of for Alice and Bob to detect a photon, and a higher
signal-to-background counts ratio we are able to improve both the
magnitude of our Bell violation as well as reduce the number of trials
required to achieve a statistically significant violation by an order of magnitude.

{\nolinenumbers
\begin{singlespace}






\noindent{\bf Acknowledgments}
We thank Carl Miller and Kevin Coakley for comments on the manuscript. A.M. acknowledges financial support through NIST grant 70NANB16H207.

\medskip

\medskip

\noindent{\bf Author Contributions}
P.B. led the project and implemented the protocol. P.B., E.K., S.G. and Y.Z. developed the protocol theory. A.M., S.J., A.R. and Y.-K. L. were responsible for extractor theory and implementation. B.C., S.W.N., M.J.S. and L.K.S. collected and interpreted the data. PB., E.K., S.G. and L.K.S. wrote the manuscript.

\medskip

\medskip

\noindent{\bf Author Information}
This work is a contribution of the National Institute of Standards and Technology and is not subject to U.S. copyright. The authors declare no competing financial interests. Correspondence and requests for materials should be addressed to P.B. (peter.bierhorst@nist.gov).

\end{singlespace}
}

\begin{figure}
\begin{center}
\includegraphics[scale=.3]{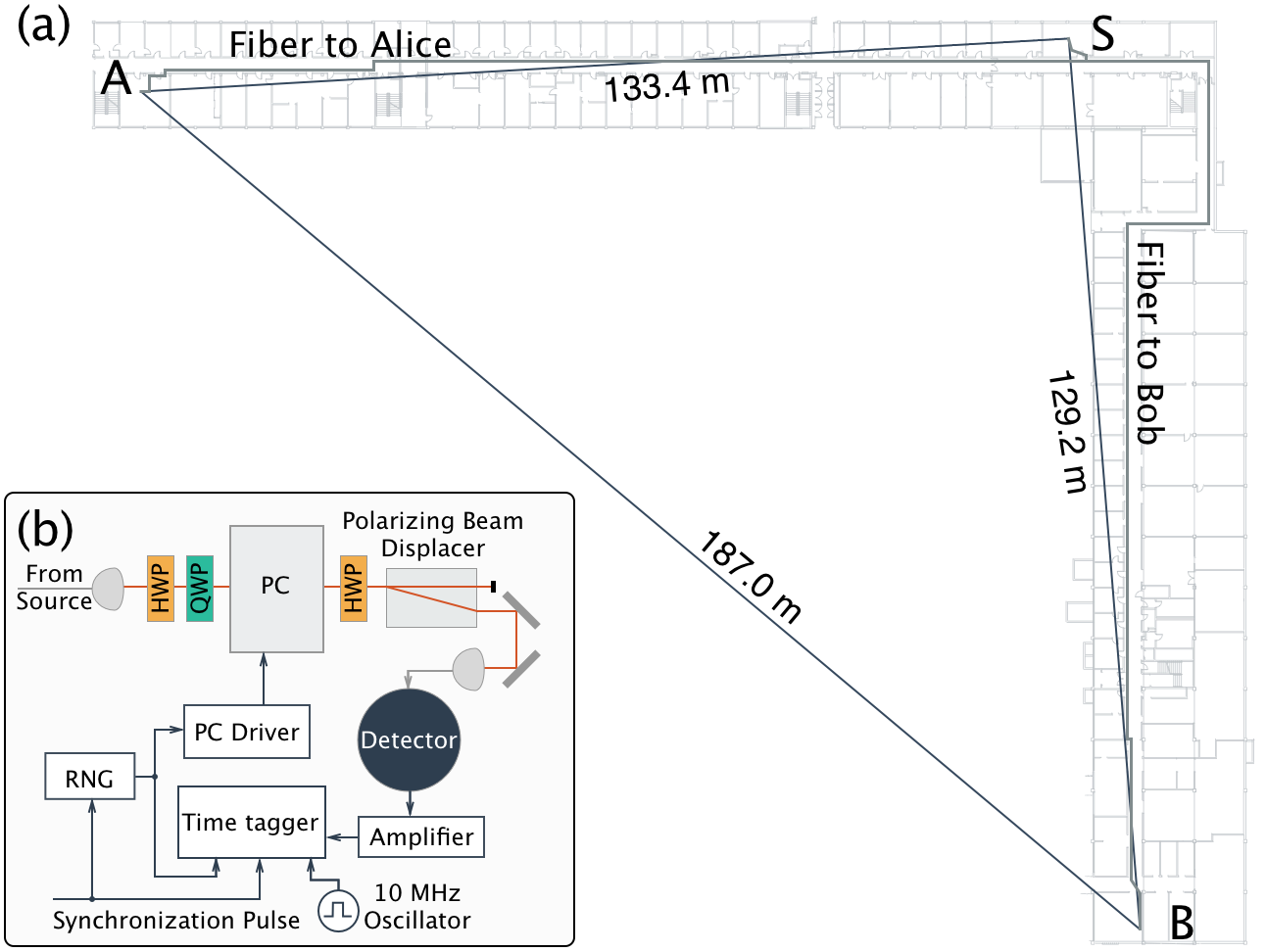}
\end{center}
\caption{{\bf The locations of the Source (S), Alice
    (A) and Bob (B)}. Each trial, the source lab
  produces a pair of photons in the non-maximally
  polarization-entangled state
  $\left|\psi \right\rangle \approx 0.982 \left|HH \right\rangle +
  0.191 \left|VV \right\rangle$, where $H$ ($V$) denotes horizontal
  (vertical) polarization. One photon is sent to Alice's lab while the
  other is sent to Bob's lab to be measured as shown in inset (b). Alice's computed optimal
  polarization measurement angles, relative to a vertical polarizer,
  are $\{a =-3.7^o, a' = 23.6^o\}$ while Bob's are
  $\{b = 3.7^o, b' = -23.6^o\}$. Both Alice and Bob use a fast Pockels cell (PC), two half-waveplates (HWP), a quarter-waveplates (QWP), and a polarizing beam displacer to
  switch between their respective polarization measurements. A
  pseudorandom number generator (RNG) governs the choice of each measurement
  setting every trial. After passing through the polarization optics, the photons are coupled into a single-mode fiber and sent to a superconducting nanowire detector. The signals from the detector are then amplified and sent to a time tagger where their arrival times are recorded and the measurement outcome is fixed. A 10 MHz oscillator keeps Alice and Bob's timetagger clocks locked. Alice and Bob are ($187 \pm 1$) m
  apart. At this distance, Alice's measurement outcome is space-like separated from the triggering of Bob's Pockels cell and vice-versa.}
\label{expschematic}
\end{figure}

\begin{figure}\centering
\includegraphics{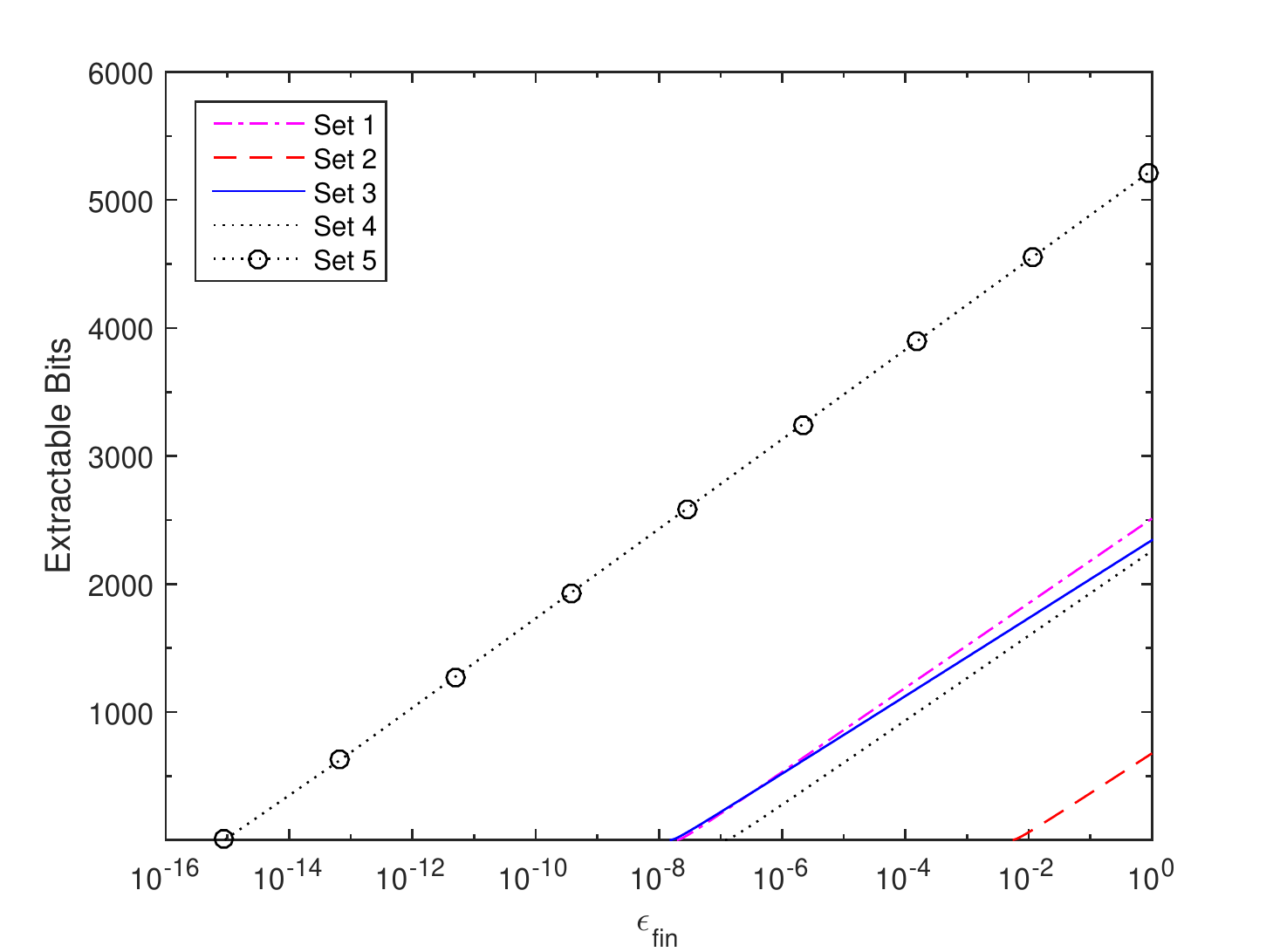}
\caption{{\bf Extractable bits as a function of error.} The figure shows the tradeoff between final error $\epsilon_{\text{fin}}$ and
  number of extractable bits $t$ for values of $\Vthresh$ pre-chosen to yield estimated passing probabilities exceeding 95\%. These thresholds were met in each case. For all data sets we set $\epsilon_{\text{p}}=\kappa^2=(0.95\,\epsilon_{\text{fin}})^2$ and $\epsilon_{\text{ext}}=0.05\,\epsilon_{\text{fin}}$, a split that was generally found to be near-optimal when numerically maximizing $t$ in \eqref{e:mttrev1} for fixed values of
$\epsilon_{\text{fin}}$.}
\label{f:FEB}
\end{figure}

\newpage

\resetlinenumber
\setcounter{page}{1}

\renewcommand{\thetable}{S\arabic{table}}   
\renewcommand{\theequation}{S\arabic{equation}}
\setcounter{table}{0}
\setcounter{equation}{0}

\begin{center}
{\Large Experimentally Generated Randomness Certified}\\
{\Large by the Impossibility of Superluminal Signals}\\
{\large (Supplementary Information)}\\
Peter Bierhorst, Emanuel Knill,  Scott Glancy, Yanbao Zhang,\\
 Alan Mink, Stephen Jordan, Andrea Rommal, Yi-Kai Liu, \\
 Bradley Christensen, Sae Woo Nam, Martin J. Stevens, Lynden K. Shalm\\
\end{center}

\def\thesection{S} 

After preliminaries to establish notation and summarize needed
properties of total variation distance and non-signaling
  distributions in \ref{st:prelim}, we give the proof of the Entropy
  Production Theorem in \ref{st:ept}. We explain how we chose the
  Bell function $T$, whose product determines whether we obtained
  the desired amount of randomness, in \ref{st:T}. We then discuss the
  parameters of the extractors obtained by the TMPS algorithm
  (\ref{st:trevisan}) and prove the Protocol Soundness Theorem
  (\ref{st:pst}). Details on how we analyzed the experimental data
  sets are in \ref{st:actual}. Justification for our claim that
  previous methods do not obtain any randomness from our low per-trial-violation data is given in \ref{st:previous}.

\subsection{Preliminaries}
\label{st:prelim}

We use the standard convention that capital letters refer to random
variables (RVs) and corresponding lowercase letters refer to values
that the RVs can take. All our RVs take values in finite sets such as
the set of bit strings of a given length or a finite subset of the
reals, so that our RVs can be viewed as functions on a finite
probability space. We usually just work with the induced joint
distributions on the sets of values assumed by the RVs.  When working
with conditional probabilities, we implicitly exclude points where the
conditioner has zero probability whenever appropriate.  We use
$\mathbb{P}(\ldots)$ to denote probabilities and $\mathbb{E}(\ldots)$
for expectations.  Inside $\mathbb{P}(\ldots)$ and when used as
conditioners, logical statements involving RVs are event
specifications to be interpreted as the event for which the statement
is true.  For example, $\mathbb{P}(R>\delta)$ is equivalent to
$\mathbb{P}(\{\omega:R(\omega)>\delta\})$, which is the probability of
the event that the RV $R$ takes a value greater than $\delta$. The
same convention applies when denoting events with $\{\ldots\}$. For
example, the event in the previous example is written as
$\{R>\delta\}$. While formally events are sets, we commonly use
logical language to describe relationships between events. For
example, the statement that $\{R>\delta\}$ implies $\{S>\epsilon\}$
means that as a set, $\{R>\delta\}$ is contained in
$\{S>\epsilon\}$. When they appear outside the the
mentioned contexts, logical statements are constraints on RVs. For
example, the statement $R>\delta$ means that all values $r$ of $R$
satisfy $r>\delta$, or equivalently, for all $\omega$,
$R(\omega)>\delta$.  As usual, comma separated statements are combined
conjunctively (with ``and'').  (In the main text, for clarity, we have
used an explicit ``AND'' for this purpose.)

If there are free RVs inside $\mathbb{P}(\ldots)$ or in the
conditioner of $\mathbb{E}(\ldots|\ldots)$ outside event specifications, the
final expression defines a new RV as a function of the free RVs.  An
example from the Entropy Production Theorem is the expression
$\mathbb{P}({\bf AB}|{\bf XY})$, which defines the RV that takes the value
$\mathbb{P}({\bf AB}={\bf ab}|{\bf XY}={\bf xy})$ when the event
$\{{\bf ABXY}={\bf abxy}\}$ occurs.  Values of RVs such as ${\bf x}$
appearing by themselves in $\mathbb{P}(\ldots)$ denote the event
$\{{\bf X}={\bf x}\}$. Thus we abbreviate expressions such as
$\mathbb{P}({\bf AB}={\bf ab}|{\bf XY}={\bf xy})$ by $\mathbb{P}({\bf ab}| {\bf xy})$.
Sometimes it is necessary to disambiguate the probability distribution
with respect to which $\mathbb{E}(\ldots)$ is to be computed.  In such cases we
use a subscript at the end of the expression consisting of a symbol
for the probability distribution, so $\mathbb{E}(T)_\mathbb{Q}$ is the expectation of
$T$ with respect to the distribution $\mathbb{Q}$. In a few instances, we use
$\llbracket\phi\rrbracket$ for logical expressions $\phi$ to denote
the $\{0,1\}$-valued function evaluating to $1$ iff $\phi$ is true.

The amount of randomness that can be extracted from an RV $R$ is
  quantified by the \textit{min-entropy}, defined as $-\log_2 \max_r
  \mathbb{P}(R=r)$. The error of the output of an extractor is given as the
  \textit{total variation} (TV) distance from uniform. Given two
  probability distributions $\mathbb{P}_1$ and $\mathbb{P}_{2}$ for $R$, the TV
  distance between them is given by
\begin{eqnarray}\label{e:condTV}
  \text{TV}(\mathbb{P}_1,\mathbb{P}_2)&=&\frac{1}{2} \sum_r \left| \mathbb{P}_1(R=r)-\mathbb{P}_2(R=r)\right|\notag\\
  &=& \sum_{r:\mathbb{P}_{1}(r)>\mathbb{P}_{2}(r)}\left( \mathbb{P}_{1}(R=r)-\mathbb{P}_{2}(R=r)\right)\notag\\
  &=& \sum_{r}\llbracket \mathbb{P}_{1}(r)>\mathbb{P}_{2}(r)\rrbracket\left( \mathbb{P}_{1}(R=r)-\mathbb{P}_{2}(R=r)\right).
\end{eqnarray}
As the name implies, the TV distance is a metric. In particular, it
satisfies the triangle inequality:
\begin{equation}\label{e:TIforTV}
\text{TV}(\mathbb{P}_{1},\mathbb{P}_{3}) \leq \text{TV}(\mathbb{P}_{1},\mathbb{P}_{2})+\text{TV}(\mathbb{P}_{2},\mathbb{P}_{3}).
\end{equation}
See Ref.~\cite{levin:2009} for this and other basic properties of TV
distances.

We sometimes compute TV distances for distributions of specific RVs,
conditional or unconditional ones. For this we introduce the notation
$\mathbb{P}_X$ for the distribution of values of $X$ according to $\mathbb{P}$, and
$\mathbb{P}_{X|Y=y}$ for the distribution of $X$ conditioned on the event
$\{Y=y\}$. With this notation, $\mathbb{P}_X\mathbb{P}_Y$ refers to the product distribution that assigns probability $\mathbb{P}_X(X=x)\mathbb{P}_Y(Y=y)$ to the event $\{X=x,Y=y\}$.

For the proof of the Protocol Soundness Theorem, we need two results
involving the TV distance. According to the first result, if $\mathbb{P}$ and $\mathbb{Q}$ are
joint distributions of RVs $V$ and $W$, where the marginals of $W$
satisfy $\mathbb{P}(w)=\mathbb{Q}(w)$, then the distance between them
is given by the average conditional distance. This is explicitly
calculated as follows:
\begin{eqnarray}\label{e:TVsameconditionals}
\text{TV}(\mathbb{P}_{VW},\mathbb{Q}_{VW}) &=&
  \sum_{w}\sum_{v}\llbracket \mathbb{P}(v,w)>\mathbb{Q}(v,w)\rrbracket
    \left(\mathbb{P}(v,w)-\mathbb{Q}(v,w)\right)\notag\\
  &=&
  \sum_{w}\sum_{v}\llbracket \mathbb{P}(v|w)\mathbb{P}(w)>\mathbb{Q}(v|w)\mathbb{Q}(w)\rrbracket
    \left(\mathbb{P}(v|w)\mathbb{P}(w)-\mathbb{Q}(v|w)\mathbb{Q}(w)\right)\notag\\
  &=& 
  \sum_{w}\sum_{v}\llbracket \mathbb{P}(v|w)>\mathbb{Q}(v|w)\rrbracket
    \left(\mathbb{P}(v|w)-\mathbb{Q}(v|w)\right)\mathbb{P}(w)\notag\\
  &=& \sum_{w}\text{TV}(\mathbb{P}_{V|W=w},\mathbb{Q}_{V|W=w})\mathbb{P}(w).
\end{eqnarray}

The second result is a special case of the data-processing inequality
for TV distance. See Ref.~\cite{pardo:1997} for this and many other
data-processing inequalities.  Let $V$ be a random variable taking
values in a finite set $\mathcal V$, and let
$F:\mathcal V \to \mathcal W$ be a function so that $F(V)$ is a random
variable taking values in the set $\mathcal W$. Then if $\mathbb{P}$
and $\mathbb{Q}$ are two distributions of $V$,
\begin{equation}\label{e:classicalprocessing}
\text{TV}\big(\mathbb{P}_V,\mathbb{Q}_V\big) \ge \text{TV}\big(\mathbb{P}_{F(V)}, \mathbb{Q}_{F(V)}\big). 
\end{equation}
Here is a proof of this inequality. Write
$\mathcal W=\{s_1,...,s_c\}$, and for each $i\in\{1,\ldots ,c\}$,
define $\mathcal V_i=\{v:f(v)=s_i\}$. The $\mathcal{V}_{i}$ form a
partition of $\mathcal V$. Then we have
\begin{eqnarray}
  \text{TV}\big(\mathbb{P}_{F(V)}, \mathbb{Q}_{F(V)}\big)&=& \frac{1}{2} \sum_{i=1}^c \left|\mathbb{P}(V\in \mathcal V_i)-\mathbb{Q}(V \in \mathcal V_i)\right|\notag\\
                                       &=& \frac{1}{2} \sum_{i=1}^c \left|\sum_{v\in \mathcal V_i}\left[\mathbb{P}(V=v)-\mathbb{Q}(V=v)\right]\right|\notag\\
                                       &\le& \frac{1}{2} \sum_{i=1}^c \sum_{v\in \mathcal V_i}\left|\mathbb{P}(V=v)-\mathbb{Q}(V=v)\right|\notag\\ 
                                       &=& \text{TV}\big(\mathbb{P}_V,\mathbb{Q}_V\big).
\end{eqnarray} 

\begin{sloppypar}
  We need to refer to the sequences of
  RVs associated with the first $i-1$ trials. To do this we use
  notation such as $({\bf AB})_{<i}$ for the outcome sequence
  $A_1B_1A_2B_2...A_{i-1}B_{i-1}$, $({\bf XY})_{<i}$ for the settings
  sequence $X_1Y_1...X_{i-1}Y_{i-1}$, and $({\bf ABXY})_{<i}$ for the
  joint outcomes and settings sequence
  $A_1B_1X_1Y_1...A_{i-1}B_{i-1}X_{i-1}Y_{i-1}$.  In general we often
  juxtapose RVs to indicate the ``joint'' RV.  From our assumption
  Eqs.~\ref{e:mtunifsettings} and~\ref{e:mtnosig} and the fact that
  $\text{past}_{i}$ subsumes the trial settings and outcomes from
  trials $1$ through $i-1$, we obtain
\begin{equation}\label{e:indeppast}
  \forall i \in (1,...,n), \quad \mathbb{P}_e\left(X_iY_i|({\bf ABXY})_{<i}\right) = \mathbb{P}_e(X_iY_i) = 1/4,
\end{equation}
and
\begin{eqnarray}\label{e:nosig}
\mathbb{P}_e(A_i|X_iY_i, ({\bf ABXY})_{<i})&=&\mathbb{P}_e(A_i|X_i, ({\bf ABXY})_{<i})\notag \\ 
\mathbb{P}_e(B_i |X_iY_i, ({\bf ABXY})_{<i})&=&\mathbb{P}_e(B_i |Y_i, ({\bf ABXY})_{<i}).
\end{eqnarray}

Eq.~\ref{e:indeppast} can be weakened to accommodate imperfect
  settings randomness by replacing it with the following two
  assumptions, where $\alpha\in [0,1/4)$ is a parameter controlling deviation from uniformity: 
\begin{eqnarray}
  \forall i \in (1,...,n), \quad 1/4 - \alpha \le \mathbb{P}_e\left(X_iY_i|({\bf ABXY})_{<i}\right) \le 1/4 + \alpha\label{e:indeppastalpha}\\
P_e(X_iY_i|(ABXY)_{<i})=P_e(X_iY_i|(XY)_{<i}) \label{e:condindep}
\end{eqnarray}
Eq.~\ref{e:indeppast} is a strictly
  stronger assumption as it implies both Eq.~\ref{e:indeppastalpha}
  (with $\alpha=0$) and Eq.~\ref{e:condindep}. Eqs.~\ref{e:nosig},
  \ref{e:indeppastalpha}, and \ref{e:condindep} are the forms of our
  assumptions used in the proof of the Entropy Production Theorem. Eq.~\ref{e:condindep} expresses conditional independence of all past outcomes and the upcoming settings given the past settings.  It
  is a special case of the Markov-chain condition in
  Ref.~\cite{dupuis:2016}.
\end{sloppypar}

For a generic trial of a two station Bell test, a distribution is
defined to be non-signaling if
\begin{equation}\label{e:nosiggen}
\mathbb{P}(A|XY)=\mathbb{P}(A|X) \quad \text{and} \quad
\mathbb{P}(B |XY)=\mathbb{P}(B |Y).
\end{equation}
Such distributions form a convex polytope and include the
\textit{local realist} (LR) distributions.  Using the conventions
of~\cite{BBP}, these are defined as follows: Let $\lambda$ range over
the set of sixteen four-element vectors of the form
$(a_0,a_1,b_0,b_1)$ with elements in $\{\text{+},0\}$.  Each $\lambda$
induces settings-conditional deterministic
distributions according to
\begin{equation}\label{e:localdet}
\mathbb{P}^\lambda(ab|xy) = \begin{cases}
1, & \text{ if $a=a_x$ and $b=b_y$,}\\
0, & \text{ otherwise.}\\
\end{cases}
\end{equation}
Then a probability distribution $\mathbb{P}$ is LR iff its
conditional probabilities $\mathbb{P}(ab|xy)$ can be written as a convex
combination of the $\mathbb{P}^\lambda(ab|xy)$. That is
\begin{equation}\label{e:local}
\mathbb{P}(ab|xy)=\sum_\lambda q_\lambda \mathbb{P}^\lambda(ab|xy),
\end{equation}
with $q_{\lambda}$ a $\lambda$-indexed set of nonnegative numbers summing to 1. 
This definition agrees with the one given in the main text. 

The eight ``Popescu-Rohrlich (PR)
  boxes'' \cite{PRBOX} are examples of non-signaling distributions
  that are not LR. One of the PR boxes is defined by
\begin{equation}\label{e:PRbox}
\mathbb{P}_{\text{PR}}(ab|xy)=\begin{cases} 1/2 & \text{ if } xy\ne 11 \text{ and }a=b, \text{ or if } xy = 11 \text{ and }a\ne b,\\
0 & \text{ otherwise,}
\end{cases}
\end{equation}
and the other seven are obtained by relabeling settings or outcomes.
We take advantage of the facts that a PR box contains one bit of
randomness conditional on the settings and that the PR boxes
together with the $16$ deterministic LR distributions of \eqref{e:localdet} form the set of
extreme points of the non-signaling polytope~\cite{barrett:2005}.

\subsection{Proof of the Entropy Production Theorem} 
\label{st:ept}

The conditions on $T$ given in the main text are that (1) $T>0$, (2)
$\mathbb{E}(T)_\mathbb{P}\leq 1$ for every LR distribution
$\mathbb{P}$, (3) there exists an $m>0$ such that
$\mathbb{E}(T)_\mathbb{Q}\leq 1+m$ for every non-signaling
distribution $\mathbb{Q}$ if the settings distribution is uniform as
in \eqref{e:mtunifsettings}, and (4) the bound $1+m$ is
achievable. Our proof of the Entropy Production Theorem does not
require that the fourth condition is satisfied. Furthermore, we prove
the Entropy Production Theorem with a weakened form of the second and
third conditions, assuming that $T$ satisfies conditions (2) and (3)
with any settings distribution satisfying \eqref{e:indeppastalpha}. In
the following, we call this relaxed version of conditions (1)-(3)
``the Bell-function conditions with bound $m$ and settings parameter
$\alpha$''. We also generalize the Entropy Production Theorem by
allowing the $T_i$ to be chosen based on $({\bf abxy})_{<i}$. We call
$T_{i}$ a ``past-parametrized family of Bell functions'' if for all
$({\bf abxy})_{<i}$, $T_{i}(a_{i}b_{i}x_{i}y_{i},({\bf abxy})_{<i})$
satisfies the Bell-function conditions with bound $m$ and settings
parameter $\alpha$ when considered as a function of the results
$a_{i}b_{i}x_{i}y_{i}$ from the $i$'th trial. By proving the theorem
for past-parametrized Bell functions $T$, we allow for the possibility
of dynamically adapting $T$ during run time, a feature that could
compensate for experimental drift in future implementations of the
protocol. The theorem and its proof can also be directly applied to
the special case where $T_i$ is the same function for all trials $i$
and $\alpha=0$.

\begin{Theorem}\label{t:ept}
  Let $T_{i}$ be a past-parametrized family of Bell functions as
  defined in the previous paragraph.  Then in an experiment of $n$
  trials obeying \eqref{e:nosig}, \eqref{e:indeppastalpha} and \eqref{e:condindep}, the
  following inequality holds for all $\epsilon_{\mathrm{p}} \in (0,1)$
  and $\Vthresh$ satisfying
    $1\le \Vthresh \le (1+(3/2)m)^{n}\epsilon_{\mathrm{p}}^{-1}$:
  \begin{equation}
    \mathbb{P}_e\left(\mathbb{P}_e({\bf AB}|{\bf XY})> \delta , V\ge \Vthresh \right) \le\epsilon_{\mathrm{p}}
  \label{e:1}
  \end{equation}
  where $\delta =
  [1+(1-\sqrt[n]{\epsilon_{\mathrm{p}}\Vthresh})/2m]^n$ and
  $\mathbb{P}_e$ represents the probability distribution conditioned on the
  event $\{E=e\}$.
\end{Theorem}

We include the constraint
$\Vthresh\leq(1+(3/2)m)^{n}\epsilon_{\text{p}}^{-1}$ for technical
reasons. Higher values of $\Vthresh$ are unreasonably large and result
in pass probabilities that are too low to be relevant. Note that this
bound ensures $\delta\ge 2^{-2n}$, a fact that will be useful in
proving the Protocol Soundness Theorem in (\ref{st:pst}).

\begin{proof}
  Since the condition on $\{E=e\}$ appears uniformly throughout, in
  this proof we omit the subscript on $\mathbb{P}_{e}$ specifying conditioning
  on $\{E=e\}$.

  The strategy of the proof is to first obtain an upper bound on the
  one-trial outcome probabilities from the expectations of Bell
  functions $T$.  This bound can be chained to give a bound on the
  probabilities of the outcome sequence as a monotonically decreasing
  function of the product of the conditional expectations of the
  $T_{i}$. That is, a larger product of expectations yields a smaller
  maximum probability and therefore more extractable randomness.  This
  product cannot be directly observed, so we relate it to the observed
  product $V$ of the $T_{i}$ via the Markov inequality applied to an
  associated positive, mean-$1$ martingale. In the following, we
  suppress the arguments $a_{i}b_{i}x_{i}y_{i}$ and
  $({\bf ABXY})_{<i}$ of $T_{i}$.
    
  The one-trial outcome probabilities are bounded by means of the
  following lemma:

  \begin{Lemma} \label{l:bound} Let $T$ satisfy the Bell-function conditions with
    bound $m>0$ and settings parameter $\alpha$.  For any non-signaling
      distribution $\mathbb{P}$ satisfying \eqref{e:indeppastalpha},
    \begin{equation}\label{e:maxprobbound} 
      \max_{abxy}\mathbb{P}(ab|xy)\le 1+
      \frac{1-\mathbb{E}[T(A,B,X,Y)]_\mathbb{P}}{2m}.
    \end{equation}
  \end{Lemma}

  \begin{proof}The settings-conditional distribution $\mathbb
      P(ab|xy)$ is non-signaling, so it can be obtained
      as a convex combination of extremal such distributions. The convex combination requires at most one PR box (\cite{bierhorst:2016}, Corollary 2.1), so we write
      $\mathbb{P}(ab|xy)=p\mathbb{Q}(ab|xy)+(1-p)\mathbb{Q}'(ab|xy)$,
      where $\mathbb{Q}$ is the PR box and $\mathbb{Q}'$ is LR.  We
      thus have
      
  \vspace{-8mm} \begin{multline}
    \mathbb{E}(T)_{\mathbb P}=\sum_{abxy} T(abxy)\mathbb{P}(abxy)= \sum_{xy} \left(\sum_{ab} T(abxy)\mathbb{P}(ab|xy)\right) \mathbb P (xy)\\
   = p\sum_{abxy}  T(abxy)\mathbb{Q}(ab|xy) \mathbb P (xy) + (1-p)\sum_{abxy}  T(abxy)\mathbb{Q}'(ab|xy) \mathbb P (xy)\\
   \le p (1+m) + (1-p)=1+pm,
     \end{multline}
     where the inequality above holds because $\mathbb{Q}(ab|xy) \mathbb P (xy)$ and $\mathbb{Q}'(ab|xy) \mathbb P (xy)$ respectively define non-signaling and LR distributions satisfying \eqref{e:indeppastalpha}, and hence these distributions respectively satisfy  $\mathbb{E}(T)\leq 1+m$ and $\mathbb{E}(T)\leq 1$. The above inequality can be re-written as $p\geq (\mathbb{E}(T)_{\mathbb{P}}-1)/m$. Now since the PR box assigns
    $xy$-conditional probability $1/2$ to at least one outcome
    different from $ab$, it follows that the $xy$-conditional
    probability relative to $\mathbb{P}$ of an outcome different from
    $ab$ is at least $p/2$. Therefore, $\mathbb{P}(ab|xy)\le 1-p/2 \le
    1-(\mathbb{E}(T)_{\mathbb{P}}-1)/(2m)$. Since $ab$ and $xy$ are
    arbitrary, this gives the inequality the lemma.
  \end{proof}
  
  We can now establish a bound on $\mathbb{P}({\bf ab}|{\bf xy})$ as follows:
  \begin{eqnarray}
    \mathbb{P}({\bf ab}|{\bf xy})
      &=& \prod_{i=1}^n\mathbb{P}(a_ib_i|({\bf ab})_{<i}, {\bf xy})\notag\\
      &=&\prod_{i=1}^n\mathbb{P}(a_ib_i|({\bf abxy})_{<i}, x_iy_i)\notag\\
      &\le& \prod_{i=1}^n\left[1+\frac{1-\mathbb{E}(T_i|({\bf abxy})_{<i})}{2m}\right].\label{e:eptstep2}
  \end{eqnarray}
  Here, the first identity is the chain rule for conditional
  probabilities, and the second follows from repeated applications of the following identity, which holds for all $j$ in $(i+1, i+2, ..., n)$ (where we recall that $({\bf xy})_{<n+1} = {\bf xy}$ and $({\bf ab})_{<i}{(\bf xy })_{<i+1} = ({\bf abxy})_{<i}, x_iy_i$). The third equality below is a consequence of \eqref{e:condindep}:
    \begin{eqnarray}\label{e:anidentity}
  \mathbb{P}(a_ib_i|({\bf ab})_{<i}, ({\bf xy})_{<j+1}) &=&\frac{  \mathbb{P}(a_ib_i,({\bf ab})_{<i}, ({\bf xy})_{<j+1})}{  \mathbb{P}(({\bf ab})_{<i}, ({\bf xy})_{<j+1})} \notag\\
  &=& \frac{  \mathbb{P}(x_{j}y_{j}|a_ib_i,({\bf ab})_{<i}, ({\bf xy})_{<j}) \mathbb{P}(a_ib_i,({\bf ab})_{<i}, ({\bf xy})_{<j})}{  \mathbb{P}(x_{j}y_{j}|({\bf ab})_{<i}, ({\bf xy})_{<j})\mathbb{P}(({\bf ab})_{<i}, ({\bf xy})_{<j})}\notag\\
  &=& \frac{  \mathbb{P}(x_{j}y_{j}| ({\bf xy})_{<j}) \mathbb{P}(a_ib_i,({\bf ab})_{<i}, ({\bf xy})_{<j})}{  \mathbb{P}(x_{j}y_{j}|({\bf xy})_{<j})\mathbb{P}(({\bf ab})_{<i}, ({\bf xy})_{<j})}\notag\\
  &=&\mathbb{P}(a_ib_i|({\bf ab})_{<i}, ({\bf xy})_{<j}).
    \end{eqnarray}
Finally, the inequality in \eqref{e:eptstep2} is a consequence of our assumption in \eqref{e:nosig} that the past-dependent distributions are non-signaling, which allows us to apply the bound from Lemma \ref{l:bound}.
      Now, by twice using
  the fact that the geometric mean of a set of positive numbers is
  always less than or equal to their arithmetic mean, we continue from
  the last line of \eqref{e:eptstep2}:
  \begin{eqnarray} 
    \prod_{i=1}^n\left[1+\frac{1-\mathbb{E}(T_i|({\bf
          abxy})_{<i})}{2m}\right] &=&
    \left(\left\{\prod_{i=1}^n\left[1+\frac{1-\mathbb{E}(T_i|({\bf
              abxy})_{<i})}{2m}\right]\right\}^\frac{1}{n}\right)^n\notag\\
    &\le&\left(\frac{\sum_{i=1}^n\left[1+\frac{1-\mathbb{E}(T_i|({\bf
              abxy})_{<i})}{2m}\right]}{n}\right)^n\notag\\
    &=&\left(1+\frac{1}{2m}-\frac{\sum_{i=1}^n\left[\frac{\mathbb{E}(T_i|({\bf
              abxy})_{<i})}{2m}\right]}{n}\right)^n\notag\\
    &\le&\left(1+\frac{1}{2m}-\left[\prod_{i=1}^n\frac{\mathbb{E}(T_i|({\bf
            abxy})_{<i})}{2m}\right]^\frac{1}{n}\right)^n\notag\\
    &=&\left(1+\frac{1-\left[\prod_{i=1}^n\mathbb{E}(T_i|({\bf
            abxy})_{<i})\right]^{\frac{1}{n}}}{2m}\right)^n. \label{e:geomean}
  \end{eqnarray}

  \begin{sloppy}
    Referring back to the statement of the theorem, we see that
    $\delta$ can be expressed as $f(\epsilon_{\text{p}}\Vthresh)$
    where $f(x)=[1+(1-\sqrt[n]{x})/2m]^n$.  Expressing
    \eqref{e:geomean} in terms of this same function $f$, we see
    that the event $\{\mathbb{P}({\bf AB}|{\bf XY})> \delta\}$ implies the
    event $\left\{f\left(\prod_{i=1}^n\mathbb{E}(T_i|({\bf ABXY})_{<i})\right)
      > \delta\right\}$. The latter event is the same as
    $\left\{\prod_{i=1}^n\mathbb{E}(T_i|({\bf ABXY})_{<i})<f^{-1}(
      \delta)=\epsilon_{\text{p}}\Vthresh\right\}$, since $f^{-1}$ is
    strictly decreasing. Conjoining the event $\{V\geq \Vthresh\}$ to
    both sides of the implication, we have $\{\mathbb{P}({\bf AB}|{\bf XY})>
    \delta,V\geq \Vthresh\}$ implies $\left\{\prod_{i=1}^n\mathbb{E}(T_i|({\bf
        ABXY})_{<i})<\epsilon_{\text{p}}\Vthresh,V\geq
      \Vthresh\right\}$, and so by the monotonicity of probabilities,
    \begin{equation}\label{e:alt1}
      \mathbb{P}\left(\mathbb{P}({\bf AB}|{\bf XY})> \delta,
        V\ge \Vthresh\right) \leq \mathbb{P}\left(\prod_{i=1}^n\mathbb{E}(T_i|(\bfABXY)_{<i})<\epsilon_{\text{p}}\Vthresh , V \geq \Vthresh\right).
    \end{equation}
    The event $\{\Phi\}$ whose probability appears on the left-hand
    side of this equation is the event in the theorem statement whose
    probability we are required to bound.  For any values of the RVs,
    the two inequalities in the event on the right-hand side imply the
    inequality in the event
    $\{\Psi\}=\left\{V/\prod_{i=1}^n\mathbb{E}(T_i|(\bfABXY)_{<i})\ge
      1/\epsilon_{\text{p}}\right\}$. Hence $\mathbb{P}(\Phi)\leq \mathbb{P}(\Psi)$.  It
    remains to show that $\mathbb{P}(\Psi)\leq \epsilon_{\text{p}}$. For this
    purpose we define the sequence $\{W_c\}_{c=1}^{n}$ of RVs by
    \begin{equation}
      W_c = \prod_{i=1}^{c}\frac{T_{i}}{\mathbb{E}(T_i|(\bfABXY)_{<i})},
    \end{equation}
    so that $\{\Psi\}=\{W_{n}\geq 1/\epsilon_{\text{p}}\}$.
    
    By definition, $W_{c}> 0$ and the factors
    $T_{i}/\mathbb{E}(T_{i}|(\bfABXY)_{<i})$ have expectation $1$ conditional
    on the past. Sequences of RVs with these properties are
    referred to as test martingales~\cite{shafer:2009} and satisfy
    that $\mathbb{E}(W_{n})=1$, which can be verified directly by induction:
    
    \begin{align}
      \mathbb{E}(W_c|({\bf ABXY})_{<c}) &=    \mathbb{E}\left(\prod_{i=1}^{c}\frac{T_{i}}{\mathbb{E}(T_i|(\bfABXY)_{<i})}\middle|({\bf ABXY})_{<c}\right)\notag\\
      &=    \mathbb{E}\left(\left(\prod_{i=1}^{c-1}\frac{ T_{i}}{\mathbb{E}(T_{i}|({\bf ABXY})_{<i})} \right)\frac{1}{\mathbb{E}(T_{c}|({\bf ABXY})_{<c})}T_{c}\middle|({\bf ABXY})_{<c}\right)\notag\\
      &=    \left(\prod_{i=1}^{c-1}\frac{ T_{i}}{\mathbb{E}(T_{i}|({\bf ABXY})_{<i})}\right) \frac{1}{\mathbb{E}(T_{c}|({\bf ABXY})_{<c})}\mathbb{E}\left(T_{c}\middle|({\bf ABXY})_{<c}\right)\notag\\
      &= W_{c-1},\label{e:alt2}
    \end{align}
    where in the second last line, we pulled out factors that are
    functions of the conditioner $(\bfABXY)_{<c}$ by applying the rule
    that if $F$ is a function of $H$, then $\mathbb{E}(FG|H)=F\mathbb{E}(G|H)$.  Taking
    the unconditional expectation of both sides of \eqref{e:alt2} and
    invoking the law of total expectation, we have
    $\mathbb{E}(W_c)=\mathbb{E}(W_{c-1})$, and so inductively, $\mathbb{E}(W_n)=\mathbb{E}(W_1)$.  Since
    $\mathbb{E}(W_{1})=1$, the claim follows. To finish the proof of the
    Entropy Production Theorem, we apply Markov's inequality to obtain
    $\mathbb{P}(W_{n}\geq 1/\epsilon_{\text{p}})\leq\epsilon_{\text{p}}$ and
    consequently $\mathbb{P}(\Phi)\leq\epsilon_{\text{p}}$.
  \end{sloppy}

\end{proof}

Now that we have proved the Entropy Production Theorem for any
past-parametrized family of Bell functions, we can justify a strategy
of setting the remaining Bell functions to $T_{i}=1$ after $\Vthresh$
is exceeded by the running product mid-protocol.  Formally, since the
running product $V_{i-1}=\prod_{i=j}^{i-1}T_{j}$ is a function of
$({\bf ABXY})_{<i}$, we can define $T_{i}=T$ conditional on
$\{V_{i-1}<\Vthresh\}$ and $T_{i}=1$ conditional on the
complement. This optional strategy can be used to eliminate the
possibility that statistical fluctuations or experimental drift could
cause $\prod_{i=1}^nT_i$ to be less than $\Vthresh$ even though the
running product exceeded $\Vthresh$ at some point prior to $n$.
  
\subsection{Choosing the Bell Function $T$}
\label{st:T} 

The Entropy Production Theorem does not indicate how to find functions
$T$ satisfying the specified conditions. We seek a high typical value
of $V=\prod_{i=1}^{n}T_{i}$, as this permits larger values of
$\Vthresh$ and consequently more extractable randomness at the same
values of $\epsilon_{\text{p}}$ and $m$. Here, we describe a procedure
for constructing a function $T$ that can be expected to perform well
if the trial results are i.i.d.~with known distribution. We estimate
the distribution from an initial portion of the run that we set aside
as training data, and in a stable experiment we expect that the trial
results' statistics are i.i.d.~to a good approximation. Note however
that the optimistic i.i.d.~assumption is only used as a heuristic to
construct $T$; once $T$ is chosen the guarantees of the Entropy
Production Theorem hold regardless of whether the trial results are
actually i.i.d. We first focus on the scenario where \eqref{e:indeppast} is assumed to hold, then show how to proceed if this is replaced with the weaker assumptions \eqref{e:indeppastalpha} and \eqref{e:condindep}.

The observed measurement outcome frequencies for training data
generally yield a weakly signaling distribution that does not exactly
satisfy the non-signaling constraints in \eqref{e:nosiggen}, due to
statistical fluctuation. Hence one can obtain an estimated
distribution by determining the maximum likelihood non-signaling
distribution for the observed measurement outcomes frequencies as
described in Ref.~\cite{zhang:2011}.  Let $N(xy)$ be the number of
training trials at setting $xy$ and $f(ab|xy)=N(ab|xy)/N(xy)$ be the
empirical frequencies of outcome $ab$ given setting $xy$. Let
$\mathbb{Q}(a,b,x,y)$ be a candidate for the probability distribution
from which these frequencies were sampled. Then up to an additive term
independent of $\mathbb{Q}$ accounting for the settings probabilities,
the log-likelihood of $f$ given $\mathbb{Q}$ is
$L(\mathbb{Q})=\sum_{a,b,x,y}N(xy)f(ab|xy)\ln(\mathbb{Q}(a,b|x,y))$. We
maximized a variant of this function to find our estimated
distribution $\mathbb{Q}(a,b,x,y)$:

\begin{align}\label{e:convexfindNS}
&\underset{\mathbb{Q}}{\text{Maximize }} \sum_{abxy}f(ab|xy)\ln \mathbb{Q}(a,b,x,y)\\
&\begin{array}{lrcll}
\!\!\text{Subject to }& \mathbb{Q}(x,y)&=&1/4 & \text{for}\quad x, y \in\{0,1\}\\
 & \mathbb{Q}(a|x,y)&=&\mathbb{Q}(a|x) & \text{for} \quad x,y \in \{0,1\}, \quad a \in \{\text{+},0\}\\
 & \mathbb{Q}(b|x,y)&=&\mathbb{Q}(b|y) & \text{for} \quad x,y \in \{0,1\}, \quad b \in \{\text{+},0\}.
\end{array}\notag
\end{align}
The first group of constraints encode our knowledge that all settings
combinations are equally likely, and the remaining constraints are
the non-signaling constraints. Note that the conditional expressions
in these constraints are equivalently expressed as linear functions of
$\mathbb{Q}(a,b,x,y)$ after using the identities $\mathbb{Q}(x,y)=1/4$.

Once the estimated distribution $\mathbb{Q}$ is obtained, we maximize
the typical values of $V$ by taking advantage of the observation that
the conditions on $T$ imply that $V^{-1}$ is a conservative $p$-value
against local realism \cite{zhang:2011}.  Such $p$-values were studied
in Ref.~\cite{zhang:2011}, which gives a general strategy, the PBR
method, for maximizing $\mathbb{E}(\ln(V))_\mathbb{Q}$. This is useful
because typical values of $V$ are close to
$\exp({n\mathbb{E}(\ln(T))_\mathbb{Q}})$: Since
$\ln(V)=\sum_{i=1}^{n}\ln(T_{i})$ is a sum of i.i.d.\ bounded terms
(given our optimistic assumption), the central limit theorem ensures
that $\ln V$ is approximately normally distributed with mean
$n\mathbb{E}(\ln(T))_\mathbb{Q}$. We therefore perform the following
optimization problem to find $T$:

\begin{align}\label{e:convexfindT}
&\underset{T}{\text{Maximize }} \mathbb{E}(\ln (T))_{\mathbb{Q}}\\
&\begin{array}{lrcll}
\!\!\text{Subject to }& \mathbb{E}(T)_{\mathbb{P}^\lambda} &\leq&1 & \forall \lambda\\
 & T(0,0,x,y) &=&1 & \forall x,y,\\
\end{array}\notag
\end{align}
where $\mathbb{P}^\lambda$ refers to the 16 conditionally deterministic LR
distributions in \eqref{e:localdet} with uniform settings distributions. This ensures that
$\mathbb{E}(T)_{\mathbb{P}_{LR}}\le 1$ for all LR distributions $\mathbb{P}_{LR}$ with uniform settings distributions. The second constraint is motivated by the fact that in our experiments, an
overwhelming fraction of the trials have no detections for both
stations.  While it is possible that a better $\mathbb{E}(\ln(T))_{\mathbb{Q}}$ can be
obtained without this constraint, we have found that the improvement
is small and likely not statistically significant given the amount of
training data used to determine the results distribution. Since the
objective functions are concave and the constraints are linear, the
optimization problems given in \eqref{e:convexfindNS} and
\eqref{e:convexfindT} are readily solved numerically with standard
tools.

Given the assumption that the trial results are i.i.d., the previous
paragraph shows that the typical values for $V$ are exponential in the
number of trials, $V = e^{-n\mathbb{E}(\ln(T))-o(n)}$.  If the
experiment is successful in showing violation of local realism,
$\mathbb{E}(\ln(T))$ is positive.  Neglecting the contribution from
$o(n)$, with $\Vthresh=e^{n\mathbb{E}(\ln(T))}$, we can bound
$-\ln(\delta)$ as
\begin{eqnarray}
  -\ln(\delta) &=& -n\ln(1+(1-(\epsilon_{\text{p}}e^{n\mathbb{E}(\ln(T))})^{1/n})/(2m))\notag\\
    &=& -n\ln(1+(1-e^{\mathbb{E}(\ln(T))+\ln(\epsilon_{\text{p}})/n})/(2m))\notag\\
     &\geq& -n(1-e^{\mathbb{E}(\ln(T))+\ln(\epsilon_{\text{p}})/n})/(2m)\notag\\
     &=& n(e^{\mathbb{E}(\ln(T))+\ln(\epsilon_{\text{p}})/n}-1)/(2m)\notag\\
     &\geq& (n\mathbb{E}(\ln(T))+\ln(\epsilon_{\text{p}}))/(2m).\label{e:yikai}
\end{eqnarray}
where we used $-\ln(1+x)\geq -x$ and $e^{x}-1\geq x$.  This shows that
asymptotically (with $\epsilon_{\text{p}}$ constant) we get at least
$\mathbb{E}(\ln(T))\log_{2}(e)/(2m)=\mathbb{E}(\log_2(T))/(2m)$ bits
of randomness per trial. For the empirical distribution obtained from
the fifth data set (``Data Set 5'') used for the protocol
according to \eqref{e:convexfindNS}, we obtain
$\mathbb{E}(\log_{2}(T))/2m=1.42\times 10^{-4}$. The bound in
Eq.~\ref{e:yikai} shows that we can get an asymptotically positive
number of bits of randomness per trial even with $\epsilon_{\text{p}}$
exponentially small in $n$.

Now we turn to the problem of finding a function satisfying the condition $\mathbb{E}(T)_{\mathbb{P}_{LR}}\le 1$ for all LR distributions $\mathbb{P}_{LR}$ with settings distribution constrained only by the weaker condition \eqref{e:indeppastalpha}, which replaces the stronger exact uniformity condition of \eqref{e:indeppast}. To do this, we show that it is sufficient to check only distributions with the extremal settings distributions where two settings have probability $1/4 + \alpha$ and the two other settings distributions have probability $1/4-\alpha$. To see why this is possible, for a fixed positive Bell function $T$, let $\mathbb P$  be an LR
  distribution whose settings distribution is constrained by \eqref{e:indeppastalpha}. Taking
  advantage of the representation in \eqref{e:local},
  
 \vspace{-8mm}   
  
  \begin{multline}
    \mathbb{E}(T)_{\mathbb P}= \sum_{abxy} T(abxy)\mathbb{P}(ab|xy)\mathbb P (xy)
   = \sum_{abxy}  T(abxy)\left(\sum_{\lambda}q_\lambda \mathbb P^\lambda(ab|xy)\right )\mathbb P(xy)\\ = \sum_{\lambda}q_\lambda\sum_{abxy}  T(abxy)\mathbb P^\lambda(ab|xy)\mathbb P(xy)\le \max_\lambda \sum_{abxy}  T(abxy) \mathbb P^\lambda(ab|xy) \mathbb P(xy),
     \end{multline} 
     so the expected value of $T$ with respect to $\mathbb P$ is
     always less than or equal to the expected value of $T$ with
     respect to a conditionally deterministic LR distribution
     $\mathbb P^\lambda$ with the same settings distribution. Since
     each deterministic LR distribution assigns conditional
     probability 1 to a single outcome $ab$ for each of the four
     setting choices $xy$, the sum
     $\sum_{abxy} T(abxy) \mathbb P^\lambda(ab|xy) \mathbb P(xy)$
     contains only four nonzero terms. Consider the two largest
     values of $T(abxy)$ and the two smallest values of $T(abxy)$
     appearing in the four nonzero terms. Note that $\sum_{abxy} T(abxy) \mathbb P^\lambda(ab|xy) \mathbb P(xy) \leq \sum_{abxy} T(abxy) \mathbb P^\lambda(ab|xy) \mathbb P^*(xy)$, where $\mathbb P^*(XY)$ is the distribution that assigns
     probability $1/4 + \alpha$ to the two settings corresponding to
     the two largest $T$, and probability $1/4-\alpha$ to the two
     settings corresponding to the two smallest $T$. Hence for any
     $T$, we can ensure that $\mathbb{E}(T)_{\mathbb P} \le 1$ holds
     for all LR distributions by checking that
     $\mathbb{E}(T)_{\mathbb P} \le 1$ holds for each conditional distribution $\mathbb P^\lambda_{AB|XY}$ coupled with each of the $i=1,\dots,\binom{4}{2}=6$ settings distributions $\mathbb S^i_{XY}$ assigning probability $1/4+\alpha$ to two settings and $1/4-\alpha$ to two other settings. This leads us to the maximization problem

 \vspace{-8mm}   
  
     \begin{align}\label{e:convexfindTalpha}
&\underset{T}{\text{Maximize }} \mathbb{E}(\ln (T))_{\mathbb{Q}}\\
&\begin{array}{lrcll}
\!\!\text{Subject to }& \mathbb{E}(T)_{\mathbb{P}^\lambda_{AB|XY}\mathbb{S}^{i}_{XY}} &\leq&1 & \forall \lambda,i\\
 & T(0,0,x,y) &=&1 & \forall x,y.\\
\end{array}\notag
\end{align}
The new problem maximizes the same objective
function as in \eqref{e:convexfindT} subject to a larger, but still
finite, number of constraints.  It can be solved numerically to find a
Bell function for the weak settings distribution.

\subsection{The TMPS Algorithm}
\label{st:trevisan}

A strong randomness extractor with parameters
$(\sigma, \epsilon,q,d,t)$ is a function
$\text{Ext}:\{0,1\}^{q}\times \{0,1\}^d \to \{0,1\}^t$ with the
property that for any random string $R$ of length $q$ and min-entropy
at least $\sigma$, and an independent, uniformly distributed seed
string $S$ of length $d$, the distribution of the concatenation
$\text{Ext}(RS)$ with S of length $t+d$ is within TV distance
$\epsilon$ of uniform. There are constructions of extractors that
extract most of the input min-entropy $\sigma$ with few seed bits. For
a review of the achievable asymptotic tradeoffs, see
Ref.~\cite{vadhan:2012}, chapter~6.  For explicit extractors that
perform well if not optimally, we used a version of Trevisan's
construction~\cite{trevisan:2001} implemented by Mauerer, Portmann and
Scholz \cite{mauerer:2012}, which we adapted\footnote{Our adapted
  source code is available at \url{https://github.com/usnistgov/libtrevisan}.} to make it functional in our environment and to
incorporate recent constructions achieving improved parameters
\cite{ma:2012}. We call this construction the TMPS algorithm. For a
fixed choice of $\sigma$, $\epsilon$ and $q$, the TMPS algorithm can
construct a strong randomness extractor for any value $t$ obeying the
following bound:
\begin{equation}\label{e:trev1}
 t+4\log_2 t \le \sigma-6 +  4\log_2(\epsilon).
\end{equation}
Given $t$, the length of the seed satisfies 
\begin{equation} \label{e:trev2} d\le w^2\cdot\max \left\{2, 1+
    \left\lceil[\log_2(t-e)-\log_2(w-e)]/[\log_2e-\log_2(e-1)]\right\rceil\right\},
\end{equation}
where $w$ is the smallest prime larger than
$2\times\lceil\log_2(4qt^2/\epsilon^2)\rceil$. We note that the TMPS
extractors are secure against classical and quantum side information
\cite{mauerer:2012}, and this security is reflected in the parameter
constraints.  Since we do not take direct advantage of this security,
it is in principle possible to improve the parameters in the Protocol
Soundness Theorem. It may
  also be possible relax the requirement of seed uniformity with more
  advanced constructions. For the purpose randomness amplification
  this is theoretically accomplished in Ref.~\cite{kessler:2017}.

  For the bound on the the number of seed bits given after the
  Protocol Soundness Theorem in the main text, we have $q=2n$ and
  $\epsilon=\epsilon_{\text{ext}}/2$. 
  Since for any
  $r$, there is a prime $w$ satisfying $r<w\leq 2r$,
  $w=O(\log(n)+\log(t/\epsilon))=O(\log(nt/\epsilon))$, where we
  pulled out exponents from the $\log$, and dropped and arbitrarily
  increased the implicit constants in front of each term to match
  summands. The coefficient of $w^{2}$ in the bound on $d$ is
  $O(\log(t))$, because of the ``minus'' sign in front of the term
  containing $w$. Multiplying gives
  $d=O(\log(t)\log(nt/\epsilon_{\mathrm{ext}})^{2})$.

\subsection{Proof of the Protocol Soundness Theorem}
\label{st:pst}

The distinction between the stations was needed to establish the
inequality in the Entropy Production Theorem and plays no further role
in this section.  We therefore simplify the notation by abbreviating
${\bf C}={\bf AB}$ and either ${\bf Z}={\bf XY}$ or
${\bf Z}={\bf XY}E$. In the former case $\mathbb{P}(\ldots)$ refers to
probabilities conditional on $\{E=e\}$. Otherwise, $\mathbb{P}(\ldots)$
involves no implicit conditions. The Protocol Soundness Theorem holds
regardless of which definition of ${\bf Z}$ is in force.  We write
$R_{\text{pass}}$ to refer to the RV that takes value $1$ conditional
on the passing event $\{V\ge \Vthresh\}$ and $0$ otherwise. The constants
$\epsilon_{\text{p}}$ and $\delta$ appearing below are the same as in
the Entropy Production Theorem.

\begin{Theorem}
  Let $0<\epsilon_{\mathrm{ext}},\kappa<1$. Suppose $\mathbb P(\mathrm{pass})\ge\kappa$, and suppose $t$ is a positive integer satisfying
  \begin{equation}\label{e:mttrev1st}
    t+4\log_2t \le -\log_2 \delta + \log_2\kappa +5\log_2 \epsilon_{\mathrm{ext}} -11.
  \end{equation}
  Then if $\text{Ext}:\{0,1\}^{2n}\times \{0,1\}^d \to \{0,1\}^t$ is
  obtained by the TMPS algorithm with parameters
  $\sigma=-\log_{2}[2\delta/(\kappa\epsilon_{\mathrm{ext}})]$ and
  $\epsilon=\epsilon_{\mathrm{ext}}/2$, and {\bf S} is a random
  bit string of length $d$ independent of the joint distribution of
  ${\bf C},{\bf Z},R_{\mathrm{pass}}$, then the joint distribution of ${\bf U}=\mathrm{Ext}({\bf CS})$, ${\bf Z}$, ${\bf S}$ and $R_{\mathrm{pass}}$ satisfies
  \begin{equation}\label{e:pst}
    \mathrm{TV}\big(\mathbb{P}_{{\bf UZS}|R_{\mathrm{pass}}=1}, \mathbb{P}^{\mathrm{unif}}_{{\bf U}}\mathbb{P}^{\mathrm{unif}}_{{\bf S}}\mathbb{P}_{{\bf Z}|R_{\mathrm{pass}}=1}\big) \le \epsilon_{\mathrm{p}}/\mathbb P(\mathrm{pass})+\epsilon_{\mathrm{ext}},
  \end{equation}
  where $\mathbb{P}^{\mathrm{unif}}$ denotes the
  uniform probability distribution.
\end{Theorem}

At this point it is tempting to just apply an extractor to ${\bf AB}$
with parameter $\sigma$ given by the nominal
$\epsilon_{\text{p}}$-smooth min-entropy
$\sigma=-\log_{2}(\delta)$. However, this does not guarantee the
strong condition \eqref{e:pst}. Specifically, there are three reasons
that \eqref{e:1} of the Entropy Production Theorem does not
immediately support the application of an extractor to ${\bf AB}$. The
first is that as specified, the extractor input should have
min-entropy $-\log_2\max_{\bf ab}\mathbb{P}({\bf AB}={\bf ab})=\sigma$
with no smoothness error. The second is that the settings-conditional
smooth min-entropies can be substantially smaller than the nominal
one.  The third is that the min-entropy is also affected by the
probability of passing being less than $1$. Accounting for
these effects requires an analysis of the settings- and
pass-conditional distributions and the extractor parameters specified
in the theorem.

\begin{proof}

The proof proceeds in two main steps inspired by the corresponding
arguments in Ref.~\cite{pironio:2013}. In the first we determine a
probability distribution $\mathbb{P}^*$ that is within
$\epsilon_{\text{p}}$ of $\mathbb{P}$ but satisfies an appropriate
bound on the conditional probabilities of ${\bf C}$ with probability
$1$ rather than $1-\epsilon_{\text{p}}$.  The distribution
$\mathbb{P}^*$'s marginals agree with those of $\mathbb{P}$ on ${\bf
  ZS}$. The probabilities conditional on aborting also agree, and
uniformity and independence of ${\bf S}$ is preserved.  In the second
step, we apply a proposition from Ref.~\cite{konig:2008} on applying
extractors to distributions such as $\mathbb{P}^{*}$ whose average
maximum conditional probabilities satisfy a specified bound. The proposition enables us to determine
the extractor parameters that achieve the required final distance
$\epsilon_{\mathrm{p}}/\mathbb P(\mathrm{pass}) +
\epsilon_{\mathrm{ext}}$ in the theorem.

  The Entropy Production Theorem guarantees that
  $\mathbb{P}(\mathbb{P}({\bf C}|{\bf Z})>\delta,
  R_{\text{pass}}=1)\leq\epsilon_{\text{p}}$. In the case where $E$ is
  included in $\bf{Z}$, this follows by the uniformity in $\{E=e\}$ of
  the theorem's conclusion:
  \begin{eqnarray}
    \mathbb{P}(\mathbb{P}({\bf C}|{\bf Z},E)>\delta, R_{\text{pass}}=1)
    &=& \sum_{e}\mathbb{P}(\mathbb{P}({\bf C}|{\bf Z},E)>\delta, R_{\text{pass}}=1|E=e)\mathbb{P}(E=e)\notag\\
    &=&\sum_{e}\mathbb{P}(\mathbb{P}({\bf C}|{\bf Z},E=e)>\delta, R_{\text{pass}}=1|E=e)\mathbb{P}(E=e)\notag\\
    &\le&\sum_{e}\epsilon_{\text{p}}\mathbb{P}(e)\notag\\
    &=&\epsilon_{\text{p}}.
  \end{eqnarray}

  Using the following construction, one may observe that for any
  random variable $U$ with values in a set of cardinality $K$ and
  $\gamma$ satisfying $1/K\le\gamma$, and any distribution
  $\mathbb{P}'$ of $U$, there exists $\mathbb{P}''$ such that
 $\mathbb{P}''(U=u)\leq\gamma$ for all possible outcomes $u$ and $\mathbb{P}''$ is within TV distance
  $\mathbb{P}'(\mathbb{P}'(U)>\gamma)$ of $\mathbb{P}'$.  To construct
  $\mathbb{P}''$, for $u$ such that $\mathbb{P}'(u)>\gamma$, set
  $\mathbb{P}''(u)=\gamma$. To compensate for the reduced
  probabilities, increase the values of $\mathbb{P}'$ to obtain those
  of $\mathbb{P}''$ without exceeding $\gamma$ on the set $\{u
  :\mathbb{P}'(u)\le\gamma\}$ so that $\mathbb{P}''$ is a normalized
  probability distribution. This is possible because in constructing
  $\mathbb{P}''$ from $\mathbb{P}'$, the total reduction in
  probability on $\{u:\mathbb{P}'(u)>\gamma\}$ given by
  $r_{-}=\sum_{u:\mathbb{P}'(u)>\gamma}(\mathbb{P}'(u)-\gamma)$ is
  less than the maximum total increase possible given by
  $r_{+}=\sum_{u:\mathbb{P}'(u)\le\gamma}(\gamma-\mathbb{P}'(u))$, as
  a consequence of $\gamma\geq 1/K$.  To see this, compute
  $r_{+}-r_{-} = \sum_{u}(\gamma-\mathbb{P}'(u))\geq
  \sum_{u}(1/K-\mathbb{P}'(u))= 0$.  The distance
  $\text{TV}(\mathbb{P}',\mathbb{P}'')$ is given by
  $\sum_{u:\mathbb{P}'(u)>\gamma}(\mathbb{P}'(u)-\gamma) \le
  \mathbb{P}'(\mathbb{P}'(U)>\gamma)$. We can now construct
  $\mathbb{P}^*$ by defining its conditional distributions on ${\bf
    C}$.  For this, substitute $U\leftarrow {\bf C}$,
  $\mathbb{P}'(U)\leftarrow \mathbb{P}({\bf C}|{\bf
    z},R_{\text{pass}}=1)$, $\gamma\leftarrow
  \delta/\mathbb{P}(R_{\text{pass}}=1|{\bf z})$ and
  $\mathbb{P}''(U)\leftarrow \mathbb{P}^*({\bf C}|{\bf
    z},R_{\text{pass}}=1)$.  The constraint on $\gamma$ is satisfied
  because the upper bound on $\Vthresh$ in the statement of the
  Entropy Production Theorem ensures that $\delta\geq2^{-2n}$.  Each
  conditional distribution satisfies $\mathbb{P}^*({\bf C}|{\bf
    z},R_{\text{pass}}=1)\leq \delta/\mathbb{P}(R_{\text{pass}}=1|{\bf
    z})$, which is equivalent to $\mathbb{P}^*({\bf
    C},R_{\text{pass}}=1|{\bf z})\leq \delta$, and is within TV
  distance $\mathbb{P}\big (\mathbb{P}({\bf C}|{\bf z},
  R_{\text{pass}}=1)>\delta/\mathbb{P}(R_{\text{pass}=1}|{\bf z})\big
  |{\bf z},R_{\text{pass}}=1 \big)$ of $\mathbb{P}_{{\bf C}|{\bf
      z},R_{\text{pass}}=1}$. The joint probability distribution
  $\mathbb{P}^*$ is determined pointwise from the already assigned
  values of $\mathbb{P}^{*}({\bf c}|{\bf z}r_{\text{pass}})$ for
  $r_{\text{pass}}=1$ as
  \begin{equation}
    \mathbb{P}^*({\bf czs}r_{\text{pass}}) =
    \left\{\begin{array}{ll}
        \mathbb{P}^*({\bf c}|{\bf z}r_{\text{pass}})
         \mathbb{P}({\bf zs}r_{\text{pass}}) & \textrm{if $r_{\text{pass}}=1$}\\
        \mathbb{P}({\bf czs}r_{\text{pass}})&\textrm{otherwise}.
      \end{array}\right.
  \end{equation}
  Since the marginal distribution of ${\bf ZS}R_{\text{pass}}$ is
  unchanged, the full TV distance between $\mathbb{P}$ and $\mathbb{P}^*$ is given by
  the average conditional TV distance with respect to
  ${\bf ZS}R_{\text{pass}}$, see \eqref{e:TVsameconditionals}.  Since
  the conditional TV distance is zero when $R_{\text{pass}}=0$ and
  from independence of ${\bf S}$, we obtain
  \begin{eqnarray}
    \text{TV}(\mathbb{P}^*_{{\bf CZS}R_{\text{pass}}},\mathbb{P}_{{\bf CZS}R_{\text{pass}}}) \hspace*{-1in}&&\notag\\
    &=&\sum_{{\bf zs}r_{\text{pass}}}
    \text{TV}\big(\mathbb{P}^*_{{\bf C}|{\bf zs}r_{\text{pass}}},
    \mathbb{P}_{{\bf C}|{\bf zs}r_{\text{pass}}}\big) \mathbb{P}({\bf zs}r_{\text{pass}})
    \notag\\
    &=&\sum_{{\bf zs}r_{\text{pass}}}
    \text{TV}\big(\mathbb{P}^*_{{\bf C}|{\bf zs}r_{\text{pass}}},
    \mathbb{P}_{{\bf C}|{\bf zs}r_{\text{pass}}}\big) \llbracket r_{\text{pass}}=1\rrbracket \mathbb{P}({\bf zs}r_{\text{pass}})
    \notag\\
    &\le& \sum_{{\bf zs}r_{\text{pass}}}
  \mathbb{P}\big(\mathbb{P}({\bf C},R_{\text{pass}}=1|{\bf z})>\delta\big |{\bf z},R_{\text{pass}}=1\big)
   \llbracket r_{\text{pass}}=1\rrbracket
   \mathbb{P}({\bf zs}r_{\text{pass}})\notag\\
    &=& \sum_{{\bf z}r_{\text{pass}}}
  \mathbb{P}\big(\mathbb{P}({\bf C},R_{\text{pass}}=1|{\bf z})>\delta\big |{\bf z},R_{\text{pass}}=1\big)
  \llbracket r_{\text{pass}}=1\rrbracket
   \mathbb{P}({\bf z}r_{\text{pass}})\notag\\
    &=& \sum_{{\bf cz}r_{\text{pass}}}
    \llbracket \mathbb{P}({\bf c}r_{\text{pass}}|{\bf z})>\delta\rrbracket
    \mathbb{P}({\bf c}|{\bf z}r_{\text{pass}})\llbracket r_{\text{pass}}=1\rrbracket
   \mathbb{P}({\bf z}r_{\text{pass}})\notag\\
&=&   \sum_{{\bf cz}r_{\text{pass}}}
    \llbracket \mathbb{P}({\bf c}r_{\text{pass}}|{\bf z})>\delta\rrbracket
   \llbracket r_{\text{pass}}=1\rrbracket  \mathbb{P}({\bf c}{\bf z}r_{\text{pass}})\notag\\
    &=& \mathbb{P}(\mathbb{P}({\bf C}R_{\text{pass}}|{\bf Z})>\delta , R_{\text{pass}}=1)\notag\\  
    &\leq&\mathbb{P}(\mathbb{P}({\bf C}|{\bf Z})>\delta , R_{\text{pass}}=1 )\notag\\
    &\leq& \epsilon_{\text{p}}.
  \end{eqnarray} 
 
  At this point we can also bound the TV distance
  conditional on passing. Since $\mathbb{P}^*(R_{\text{pass}}) =
  \mathbb{P}(R_{\text{pass}})$, we can apply \eqref{e:TVsameconditionals} and
  the above bound on the distance to get
  \begin{eqnarray}
    \epsilon_{\text{p}} &\geq& 
    \text{TV}\big(\mathbb{P}^*_{{\bf CZS}R_{\text{pass}}},\mathbb{P}_{{\bf CZS}R_{\text{pass}}}\big)\notag\\
    &=&
    \sum_{r}  \text{TV}\big(\mathbb{P}^*_{{\bf CZS}|R_{\text{pass}}=r},\mathbb{P}_{{\bf CZS}|R_{\text{pass}}=r}\big)\mathbb{P}(R_{\text{pass}}=r)  \notag\\
    &=&
    \text{TV}\big(\mathbb{P}^*_{{\bf CZS}|R_{\text{pass}}=1},\mathbb{P}_{{\bf CZS}|R_{\text{pass}}=1}\big)\mathbb{P}(R_{\text{pass}}=1).
  \end{eqnarray}
  We conclude that
  \begin{equation}\label{e:step1conclusion}
    \text{TV}\big(\mathbb{P}^*_{{\bf CZS}|R_{\text{pass}}=1},\mathbb{P}_{{\bf CZS}|R_{\text{pass}}=1}\big)\leq \epsilon_{\text{p}}/\mathbb{P}(R_{\text{pass}}=1).
  \end{equation}

  For the second main step, we need the average ``guessing
  probability'' of ${\bf C}$ given ${\bf Z}$ conditional on
  $\{R_{\text{pass}}=1\}$.  This is given by
  \begin{eqnarray}\label{e:yanbao}
    \sum_{{\bf z}} \max_{{\bf c}}(\mathbb{P}^*({\bf c}|{\bf z},R_{\text{pass}}=1))\mathbb{P}({\bf z}|R_{\text{pass}}=1) &\le&
    \sum_{{\bf z}} \frac{\delta}{\mathbb{P}(R_{\text{pass}}=1|{\bf z})}\mathbb{P}({\bf z}|R_{\text{pass}}=1) \notag\\
    &=& \delta \sum_{{\bf z}}\frac{\mathbb{P}({\bf z})}{\mathbb{P}(R_{\text{pass}}=1)} \notag \\
    &\leq& \delta/\kappa.
  \end{eqnarray}
  We remark that here it is necessary to assume
    the lower bound $\kappa $ on $\mathbb P (R_{\text{pass}}=1)$ in
    order to proceed; otherwise the bound in \eqref{e:yanbao} would
    become unbounded due to potentially arbitrarily small values of $\mathbb P (R_{\text{pass}}=1)$. Now
  we can apply Proposition 1 of Ref.~\cite{konig:2008}. The next lemma
  extracts the conclusion of this proposition in the form we need. It
  is obtained by substituting the variables and expressions in the
  reference as follows: $X\leftarrow {\bf C}$, $Y\leftarrow {\bf S}$,
  $E\leftarrow {\bf Z}$, $\mathsf{E}(X,Y) \leftarrow \text{Ext}({\bf
    CS})$, $k\leftarrow -\log_{2}(\delta/\kappa)
  -\log_{2}(2/\epsilon_{\text{ext}})$, $\epsilon\leftarrow
  \epsilon_{\text{ext}}/2$ and the distributions are replaced with the
  corresponding ones that are conditional on $\{R_{\text{pass}}=1\}$.
  The guessing entropy in the reference is the negative logarithm of
  the the average guessing probability in \eqref{e:yanbao}.

  \begin{Lemma} 
    \begin{sloppy}
      Suppose that $\mathrm{Ext}$ is a strong extractor with
      parameters
      $(-\log_{2}(2\delta/(\kappa\epsilon_{\mathrm{ext}})),\epsilon_{\mathrm{ext}}/2,2n,d,t)$. Write
      ${\bf U}=\mathrm{Ext}({\bf CS})$.  Then we have the following
      bound:
      \begin{equation}
        TV\big (\mathbb{P}^*_{{\bf UZS}|R_{\mathrm{pass}}=1},         \mathbb{P}^{\mathrm{unif}}_{{\bf U}} \mathbb{P}_{{\bf S}} \mathbb{P}^*_{{\bf Z}|R_{\mathrm{pass}}=1}\big) \leq \epsilon_{\mathrm{ext}}. \label{e:lemmapm}
      \end{equation}
    \end{sloppy}
  \end{Lemma}
  
  To apply the lemma, we obtain $\text{Ext}$ by the TMPS algorithm
  with the parameters in the lemma. Expanding the logarithms as
  $\sigma=-\log_{2}(\delta)+\log_{2}(\kappa) +
  \log_{2}(\epsilon_{\text{ext}})-1$ and
  $\log_{2}(\epsilon)=\log_{2}(\epsilon_{\text{ext}})-1$ and
  substituting in Eq.~\ref{e:trev1} gives the requirement
  \begin{equation}
    t+4\log_2 t\le -\log_{2}(\delta) + \log_{2}(\kappa) + 5\log_2(\epsilon_{\text{ext}})-11,
  \end{equation}
  as asserted in the Protocol Soundness Theorem.
  The number of seed bits $d$ is obtained from Eq.~\ref{e:trev2}.

  It remains to determine the overall TV distance conditional on
  passing. Applying \eqref{e:classicalprocessing} with
  $V = {\bf C, Z, S }$ and $F$ defined as
  $F({\bf C, Z, S})=\big(\text{Ext}({\bf C,S}), {\bf Z}, {\bf
    S}\big)$, and applying \eqref{e:step1conclusion}, we have
\begin{equation}\label{e:uzs}
\text{TV}\big(\mathbb{P}^*_{{\bf UZS}|R_{\text{pass}}=1},\mathbb{P}_{{\bf UZS}|R_{\text{pass}}=1}\big)\leq
\text{TV}\big(\mathbb{P}^*_{{\bf CZS}|R_{\text{pass}}=1},\mathbb{P}_{{\bf CZS}|R_{\text{pass}}=1}\big)\leq \epsilon_{\mathrm{p}}/\mathbb P (R_{\text{pass}}=1).
\end{equation}
Then by \eqref{e:TIforTV}, \eqref{e:lemmapm} and \eqref{e:uzs}, we have
\begin{equation}\label{e:wehave}
  \text{TV}\big(\mathbb{P}_{{\bf UZS}|R_{\text{pass}}=1}, \mathbb{P}^{\text{unif}}_{{\bf U}}\mathbb{P}^{\text{unif}}_{{\bf S}}\mathbb{P}^*_{{\bf Z}|R_{\text{pass}}=1}\big)
  \le \epsilon_{\text{ext}} +\epsilon_{\mathrm{p}}/\mathbb P (R_{\text{pass}}=1).
\end{equation}
As $\mathbb{P}^*_{{\bf Z}|R_{\text{pass}}=1}=\mathbb{P}_{{\bf
    Z}|R_{\text{pass}}=1}$, the statement of the theorem
follows.\end{proof}

As discussed in the main text, the Protocol Soundness Theorem implies that the unconditional TV distance from an ``ideal protocol'' can be bounded by $\max (\epsilon_{\mathrm{p}} +
  \epsilon_{\mathrm{ext}}, \kappa)$. This error parameter is closely related to
  the security definitions appearing in, for instance,
  Equation (1) of \cite{portmann:2014} and Definition 4 of
  \cite{arnon:2016}. To explain how we arrive at $\max (
  \epsilon_{\mathrm{p}} + \epsilon_{\mathrm{ext}}, \kappa)$, note that an ideal protocol may abort with positive probability, but conditioned on not aborting it produces perfectly uniform output independent of side
  information. That is, the
  distribution of an ideal protocol $\mathbb P^{\mathrm{ideal}}_{{\bf
      UZS}R_{\text{pass}}}$ must satisfy $\mathbb
  P^{\mathrm{ideal}}_{{\bf UZS}|R_{\text{pass}}=1} =
  \mathbb{P}^{\text{unif}}_{{\bf U}}\mathbb{P}^{\text{unif}}_{{\bf
      S}}\mathbb{P}^{\mathrm{ideal}}_{{\bf Z}|R_{\text{pass}}=1}$, but
  the distribution of the ideal protocol is otherwise unconstrained
  when $R_{\text{pass}}=0$. Given our actual protocol distribution
  $\mathbb P$ we can define a particular ideal distribution with the
  same probability of passing as the actual protocol by setting
  $\mathbb{P}^{\mathrm{ideal}}_{{\bf UZS}|R_{\text{pass}}=1} =
  \mathbb{P}^{\text{unif}}_{{\bf U}}\mathbb{P}^{\text{unif}}_{{\bf
      S}}\mathbb{P}_{{\bf Z}|R_{\text{pass}}=1}$,
  $\mathbb{P}^{\mathrm{ideal}}_{{\bf UZS}|R_{\text{pass}}=0} =
  \mathbb{P}_{{\bf UZS}|R_{\text{pass}}=0}$, and
  $\mathbb{P}^{\mathrm{ideal}}(R_{\text{pass}}=1) =
  \mathbb{P}(R_{\text{pass}}=1)$. If $\mathbb{P}(R_{\text{pass}}=1)\ge
  \kappa$, the unconditional TV distance from $\mathbb{P}$ to this
  ideal protocol can be bounded by
\begin{eqnarray}
TV(\mathbb P_{{\bf UZS}R_{\text{pass}}}, \mathbb P^{\mathrm{ideal}}_{{\bf UZS}R_{\text{pass}}})&=&
\sum_{r=0,1} TV(\mathbb P_{{\bf UZS}|R_{\text{pass}}=r}, \mathbb P^{\mathrm{ideal}}_{{\bf UZS}|R_{\text{pass}}=r})\mathbb P (R_{\text{pass}}=r)\notag\\
&=&TV(\mathbb P_{{\bf UZS}|R_{\text{pass}}=1}, \mathbb P^{\mathrm{ideal}}_{{\bf UZS}|R_{\text{pass}}=1})\mathbb P (R_{\text{pass}}=1)\notag\\
&\le&\left[\epsilon_{\mathrm{p}}/\mathbb P(R_{\text{pass}}=1) + \epsilon_{\mathrm{ext}} \right]\mathbb P (R_{\text{pass}}=1)\notag\\
&\le& \epsilon_{\mathrm{p}}+ \epsilon_{\mathrm{ext}},
\end{eqnarray}
where above we used, in order, \eqref{e:TVsameconditionals}, $\mathbb{P}^{\mathrm{ideal}}_{{\bf UZS}|R_{\text{pass}}=0} =
  \mathbb{P}_{{\bf UZS}|R_{\text{pass}}=0}$, \eqref{e:pst}, and $\mathbb {P} ( R_{\mathrm{pass}}=1)\le 1$. Alternatively, if $\mathbb{P}(R_{\text{pass}}=1)< \kappa$, we have
\begin{eqnarray}
TV(\mathbb P_{{\bf UZS}R_{\text{pass}}}, \mathbb P^{\mathrm{ideal}}_{{\bf UZS}R_{\text{pass}}})
&=&TV(\mathbb P_{{\bf UZS}|R_{\text{pass}}=1}, \mathbb P^{\mathrm{ideal}}_{{\bf UZS}|R_{\text{pass}}=1})\mathbb P (R_{\text{pass}}=1)\notag\\
&\le&1\cdot \kappa \notag\\
&=&\kappa,
\end{eqnarray}
as the TV distance can never be greater than one. Thus we see that the distance from the ideal protocol is bounded by $\max (\epsilon_{\mathrm{p}}+ \epsilon_{\mathrm{ext}}, \kappa)$. However, as noted in the main text, we considered a more conservative overall error parameter $\epsilon_{\mathrm{fin}}=\max(\epsilon_{\mathrm{p}}/\kappa+ \epsilon_{\mathrm{ext}},\kappa)$. This ensures that for all pass probabilities exceeding $\kappa$, the pass-conditional distribution of the output is within $\epsilon_{\mathrm{p}}/\mathbb{P}(\mathrm{pass}) + \epsilon_{\mathrm{ext}} \le \epsilon_{\mathrm{p}}/\kappa+ \epsilon_{\mathrm{ext}}\le\epsilon_{\mathrm{fin}}$ of $\mathbb{P}^{\text{unif}}_{{\bf U}}\mathbb{P}^{\text{unif}}_{{\bf S}}\mathbb{P}_{{\bf Z}|R_{\text{pass}}=1}$.

\subsection{Protocol Application Details}
\label{st:actual}

The Protocol Soundness Theorem supports the protocol given in Table
\ref{t:protocol}, with overall soundness error given by
$\epsilon_{\mathrm{fin}} = \max(\epsilon_{\mathrm
  {p}}/\kappa+\epsilon_{\mathrm{ext}}, \kappa)$. A protocol is
furthermore {\it complete} if there exist real-world systems that pass
the protocol with reasonably high probability. The completeness of our
protocol is supported by quantum mechanics, which predicts
experimental distributions that violate nontrivial
  Bell inequalities\cite{BELL} and pass the protocol with high
probability. Completeness is also witnessed by our repeated successful
implementations of the protocol.

\begin{table}[h]\centering
\caption{\label{t:protocol}Protocol for Randomness Generation}
\begin{tabular}{ lccccc }
 \hline
\multicolumn{6}{|p{15cm}|}{ \rule{0pt}{3ex} 1. Choose a Bell function $T$ satisfying the conditions of the Entropy Production Theorem, a number of trials $n$ to be run, a threshold for passing $\Vthresh>1$, error parameters $\epsilon_{\mathrm p}, \epsilon_{\mathrm{ext}},\kappa>0$, and a positive integer $t$ for which \eqref{e:mttrev1} is satisfied.}\\
\multicolumn{6}{|p{15cm}|}{ \rule{0pt}{3ex} 2. (Entropy Production) Run a succession of $n$ experimental trials, where in each trial $i$ Alice and Bob randomly and uniformly choose respective settings $X_i,Y_i \in \{0,1\}$, and record respective outputs $A_i,B_i\in\{\text{+},0\}$. (Optional) Calculate $\prod_{j=1}^i T(A_j,B_j,X_j,Y_j)$ after each trial and re-set $T$ to the constant function $1$ for the remainder of the experiment if $\prod_{j=1}^i T(A_j,B_j,X_j,Y_j)>\Vthresh$.} \\
\multicolumn{6}{|p{15cm}|}{ \rule{0pt}{3ex} 3. Compute $\prod_{i=1}^n T(A_i,B_i,X_i,Y_i)$ and abort if this quantity does not exceed $\Vthresh$.} \\
\multicolumn{6}{|p{15cm}|}{ \rule{0pt}{3ex} 4. (Extraction) Generate a random and uniform $d$-bit seed string $\bf S$ where $d$ is given by \eqref{e:trev2} with $q=2n, \epsilon=\epsilon_{\mathrm{ext}}/2$. Output ${\bf U} = \text{Ext}({{\bf AB},{\bf S}})$ with the security guarantee given by \eqref{e:mtfinal}.}\\
 \hline
\end{tabular}
\end{table}

The five new data sets reported in the main paper were taken in
2017. Each trial in a data set encompassed fourteen time intervals, and in a given
trial, the outcome ``$+$'' was recorded if there was a detection in
any one of these intervals and ``$0$'' otherwise. The number of intervals was
 fixed and chosen in advance of running the
protocol. The five data sets
were analyzed in the order in which they were taken. We determined the
Bell function $T$ from training data consisting of the first $5\times
10^{6}$ trials as explained in \ref{st:T}.  We chose $5\times 10^6$
trials so that we could obtain a Bell function $T$ using an
accurate estimate of the experimental distribution of measurement
outcomes without sacrificing too much data that could be used for
randomness generation. After the protocol was
officially run on a data set, the same data set was re-analyzed using
different lengths of training portions to see if a different length
should be used for subsequent data sets, but there was never clear
evidence to suggest that we should
have used a different length for the training portion.
  
After training, we inferred an expected value $n\mu$ and variance
$n\sigma^{2}$ of $\sum_{i=1}^{n}\ln(T_{i})$ on the remaining trials
assuming i.i.d.\ trials and Gaussian statistics according to the
central limit theorem, where $n$ and $\mu$ were calculated according
to the distribution obtained from the optimization problem of
\eqref{e:convexfindNS}. Note that under these assumptions, we treat
$\sum_{i=1}^{n}\ln(T_{i})$ as if it were a sum of independent and
bounded RVs. Since $V = \exp\left(\sum_{i=1}^{n}\ln (T_{i})\right)$ we
can then choose $\Vthresh$ so that it has a $0.95$ chance of being
exceeded according to the Gaussian approximation, by setting
$\Vthresh=e^{n\mu-1.645\sqrt{n}\sigma}$. For Data Sets 3, 4 and 5, $\Vthresh$ was
chosen to be smaller than this value to increase the chance of passing
the protocol while still meeting desirable benchmarks for extractable
randomness.
 
We now discuss our application of the protocol to Data Set 5, and then summarize the main results for all five data sets in Table \ref{t:summary}. Data Set 5 set consists of 60,110,210 trials, roughly twice as long as each of the first four data sets. The counts for each trial outcome from the first $5\times 10^{6}$ trials are shown in Table \ref{t:rawcounts}.  The maximum likelihood non-signaling distribution
corresponding to these counts is shown in Table~\ref{t:nosig}. We
determined $T$ from this distribution, the values of $T$ are shown in
Table~\ref{t:PBR} of the main text.

  \begin{table}[h]\centering\caption{Result counts for the first $5\times
        10^{6}$ trials of Data Set 5.}\label{t:rawcounts}
 \begin{tabular}{ r|c|c|c|c| }
 \multicolumn{1}{r}{}
  &  \multicolumn{1}{c}{$ab=\text{++}$}
 &  \multicolumn{1}{c}{$ab=\text{+}0$}
 &  \multicolumn{1}{c}{$ab=0\text{+}$} 
   &  \multicolumn{1}{c}{$ab=00$}
\\
  \cline{2-5}
   $xy=00$&3166 & 1851 & 2043 & 1243520 \\
 \cline{2-5}
   $xy=01$& 3637 & 1338 & 13544 & 1230633 \\
 \cline{2-5}
   $xy=10$& 3992 & 13752 & 1226 & 1230686 \\
 \cline{2-5}
   $xy=11$&357 & 17648 & 16841 & 1215766 \\
 \cline{2-5}
 \end{tabular}
\end{table} 
 \begin{table}[h]\centering\caption{Maximum likelihood non-signaling distribution according to
the counts in Table~\ref{t:rawcounts}, rounded to eight decimal places.}\label{t:nosig}
 \begin{tabular}{ r|c|c|c|c| }
 \multicolumn{1}{r}{}
  &  \multicolumn{1}{c}{$ab=\text{++}$}
 &  \multicolumn{1}{c}{$ab=\text{+}0$}
 &  \multicolumn{1}{c}{$ab=0\text{+}$} 
   &  \multicolumn{1}{c}{$ab=00$}
\\
  \cline{2-5}
   $xy=00$&  0.00063301  &  0.00036794   & 0.00041085  &  0.24858820 \\
 \cline{2-5}
   $xy=01$&  0.00073159  &  0.00026936  &  0.00270824  &  0.24629081 \\
 \cline{2-5}
   $xy=10$&  0.00080002  &  0.00277179  &  0.00024384  &  0.24618435 \\ 
 \cline{2-5}
   $xy=11$&  0.00007087   & 0.00350093  &  0.00336896  &  0.24305924 \\ 
 \cline{2-5}
 \end{tabular}
\end{table} 

The $0.95$ rule for determining $\Vthresh$ given that there are 55,110,210 trials for the protocol yields $\Vthresh=8.79\times10^{36}$. We chose a more conservative value of $\Vthresh=1.5\times10^{32}$ to improve the odds of passing the protocol, while still allowing for the extraction of 1024 bits uniform to within $10^{-12}$. This threshold corresponds to a probability of passing of roughly 0.9916 according to the i.i.d. scenario described above. Running the protocol, this threshold was exceeded, with a final value of $V=2.018\times 10^{41}$.

The running product $\prod_{i=1}^c T_{i}$ first exceeded $\Vthresh$ at
trial number $c=41,243,976$, and one has the option of setting the
remaining $T_i=1$ regardless of outcome for the rest of the data
run. The soundness of this procedure is justified by the adaptive
properties of the Entropy Production Theorem. In our application of
the protocol, we implemented a similar strategy without technically
changing the Bell function, by relabeling all outcomes to $0$ starting
at trial number $c+1$. This also results in $T_i=1$ for the remainder
of the experiment. This strategy is justified as our assumptions allow
for Alice and Bob to cooperatively make arbitrary changes to the
experiment in advance of a trial based on the past, which includes the
current running product. Turning off the detectors to guarantee
outcomes of $0$ is one such change, and in principle there was
sufficient time (at least $5\,\mu s$) for the necessary communication
to take place after the previous trial.

Throughout, we did not consider the length $d$ of the seed in making
our choices and determined $d$ from the other parameters according to
\eqref{e:trev2}. For applying the extractor to Data Set 5, we
used 315,844 seed bits. The seed bits were
collected from one of the random number generators used to select the settings in \cite{shalm:2015}. Specifically, each seed bit came from the XOR of two bits generated by the photon-sampling random number generator described in \cite{shalm:2015}.
It took 317 seconds for our computer to construct the extractor
according to the TMPS algorithm and generate the explicit final output
string. Here is the final output string that results from applying the extractor to the string ${\bf AB}$, when ${\bf AB}$ is obtained with relabeling of all outcomes to 0 starting at trial number $41,243,977$ (after $\Vthresh$ is exceeded by the running product). 

\begin{center}
\begin{tiny}
\begin{singlespace}
11100010011111111101001100001111100101010101001101111001111010110101101000011011000111010001101000111010011110011100101101100100 

10111111111001100010110010110111101100101111010011001101101111010100111001011010111111011110010100110001000101011000001111111101 

11011001110001111100010010011100011100000000010110010101101111001011001001000001101110110000000111110111001110001100101110001100

10110110001100011101001001001010101000100001010101001001011101010101001010100111001101001010001010100001101111110110011011110000

11100110100110010111001011000110010100101000110101100100000110111000101101001101110110111111001110110011100000001111001111101100

10110000111110011100110111110110101111000001010001010110100010011101011000001001011100010110101101111100110100001110101110110101

10001010011111011110111001000001000110111111110011101001110100111000000100101100010011101110100001110101111001001011111111001100

01111011101001101010101100010010000011111110010101011010111111100011110110001010111011000001111000011111101100100010001001000010
\end{singlespace}
\end{tiny}
\end{center}

 \begin{table}[h]\centering\caption{Summary of application of protocol to data sets. For fixed goal choices of $\epsilon_{\mathrm{fin}}$, the error parameters were computed according to the formula $\epsilon_{\text{p}}=\kappa^2=(0.95\,\epsilon_{\text{fin}})^2$, $\epsilon_{\text{ext}}=0.05\,\epsilon_{\text{fin}}$. Error parameters were chosen in advance of running the protocol for Data Sets 3, 4 and 5; the $\epsilon_{\mathrm{fin}}$ and $t$ values for Data Sets 1 and 2 are marked with an asterisk as they were not chosen in advance and are only included for illustrative purposes. We remark that the quantity $1/\Vthresh$ can also be interpreted as a p-value against local realism \cite{zhang:2011}.}\label{t:summary}
 \begin{tabular}{ r|ccccccc }
   Data Set & n & m &  95\% cut off& $\Vthresh$ & $\epsilon_{\mathrm{fin}}$ & $t$ &$V>\Vthresh$ \\
\hline
   $1$ &  24865320  & 0.01066 &  $4.68\times 10^{16}$  &  $4.68\times 10^{16}$& $10^{-6*}$ & $512^{*}$ & Yes\\
   $2$ &  24809970  & 0.01126 &  $1.30\times 10^{5}$  &  $1.30\times 10^{5}$& $0.01^{*}$ & $61^*$ & Yes\\
   $3$ &  24818959  & 0.01163 &  $9.74\times 10^{19}$  &  $10^{17}$& $10^{-6}$ & 512 &Yes\\
   $4$ &  24846822  & 0.01063 &  $6.57\times 10^{15}$  &  $10^{15}$& $10^{-6}$ & 256 &Yes\\ 
   $5$ &  55110210  & 0.01004 &  $8.79\times 10^{36}$  &  $1.5\times10^{32}$& $10^{-12}$ & 1024 &Yes\\ 
 \end{tabular}
\end{table}

\begin{table}\centering\caption{2-tail p-values for consistency checks}\label{t:consistencychecks}
 \begin{tabular}{ r|ccccccc }
   Data Set &    Sig. 1 &Sig. 2 &Sig. 3 &Sig. 4 \\
\hline
  Data Set 1 & 0.507 & 0.777 &0.290 &0.323 \\
  Data Set 2 & 0.765 & 0.965 &0.115 &0.684    \\
  Data Set 3 & 0.633 & 0.072 &0.381 &0.099\\
  Data Set 4 & 0.144 & 0.320 &0.844 &0.356  \\
  Data Set 5 &  0.879 &  0.131 & 0.554 & 0.885 \\
 \end{tabular}
\end{table}

After the protocol was run, we ran consistency checks on the data sets to look for potential inconsistencies with \eqref{e:mtnosig}, the no-signaling assumption. 
Using the tests described in Ref.~\cite{shalm:2015}, we examined the four signaling equalities: 1: $\mathbb{P}(A|X=0,Y)=\mathbb{P}(A|X=0)$, 2: $\mathbb{P}(A|X=1,Y)=\mathbb{P}(A|X=1)$, 3: $\mathbb{P}(B|X,Y=0)=\mathbb{P}(B|Y=0)$, and 4: $\mathbb{P}(B|X,Y=1)=\mathbb{P}(B|Y=1)$. For these tests we used statistics whose asymptotic distributions would approach the standard normal with mean $0$ and variance $1$, if the trials were i.i.d. We report the p-values obtained from these tests for all data sets in Table \ref{t:consistencychecks}, which do not suggest any inconsistencies.

Prior to the analysis of the five data sets reported in the main text,
the protocol was applied to data sets taken as part of the experiment
reported in Ref.~\cite{shalm:2015}. These results are described in
\cite{randomv1}. After setting aside the first $5\times 10^7$ trials
of the data set XOR 3 as a training set to construct the function $T$
and choose a threshhold $\Vthresh$ based on the $95\%$ rule, the
protocol was applied to the rest of the data set with parameters
$\epsilon_{\mathrm{p}} = 3.1797 \times 10^{-4}$ and
$\epsilon_{\mathrm{ext}}=3.533 \times 10^{-5}$, which were chosen to
minimize $\epsilon_{\mathrm{p}}/\kappa + \epsilon_{\mathrm{ext}}$ for
$\kappa= 1/3$ while satisfying \eqref{e:mttrev1}. This choice of
parameters was suboptimal for minimizing either
$\epsilon_{\mathrm{fin}}$ or $\max (\epsilon_{\mathrm{p}}+
\epsilon_{\mathrm{ext}},\kappa )$, the two figures of merit disucssed
in the main text. However, the instance of the TMPS algorithm induced
by the above choice of parameters would have been induced by other
choices of parameters that perform better according to these figures
of merit. The same extraction is induced by $\epsilon_{\mathrm{p}} =
3.6509\times 10^{-4}$, $\epsilon_{\mathrm{ext}}=3.5330 \times
10^{-5}$, and $\kappa = 4.0042\times 10^{-4}$, which leads to a
distance of $\max (\epsilon_{\mathrm{p}}+
\epsilon_{\mathrm{ext}},\kappa )=4.0042\times 10^{-4}$ from an ideal
protocol for the extraction of 256 bits. We can also choose
$\epsilon_{\mathrm{p}} = 3.370\times 10^{-4}$,
$\epsilon_{\mathrm{ext}}=3.533 \times 10^{-5}$, and $\kappa = 0.0184$
to induce the same extraction with an $\epsilon_{\mathrm{fin}}$
parameter of $0.0184$. \ignore{Semi-permanent
  elaboration of what's going on here: Alan was told ``$\sigma$'' and
  ``$\epsilon$'' values based on which he extracted XOR 3; these
  figures were 357.17974717821 and
  $0.000017665=\epsilon_{\mathrm{ext}}/2$, respectively. Any other
  parameters that yield these same values will yield the same
  extractor function, so we can assess the extraction based on other
  parameter splits. This exercise entailed ``locking in''
  $\epsilon_{\mathrm{ext}}$ with the original choice, while $\kappa$
  and $\epsilon_{\mathrm{p}}$ could be varied so long as they gave at
  least the earlier $\sigma$ value ($\sigma$ is a function of $\kappa$
  and $\epsilon_{\mathrm{p}}$, as well as other fixed parameters like
  $\Vthresh$ and $m$). As a side note, if you try to optimize the
  figures of merit without locking in $\epsilon_{\mathrm{ext}}$, I was
  able to get as low as $3.9943\times 10^{-4}$ for ``distance from
  ideal'' and, in a separate parameter optimization,
  $\epsilon_{\mathrm{fin}} = 0.01536$. So, luckily, our original
  parameter choices did not constrain us too badly on either
  front.}

Statistically significant settings nonuniformity was detected for some
of the sets examined in \cite{randomv1}. This was consistent with the
finding in \cite{shalm:2015} that a combination of uncontrolled
environmental variables and the synchronization electronics introduced
small biases in the settings. This effect is not present in the 2017 data sets, which used a
reliable pseudorandom source for settings randomness. As the Entropy
Production Theorem can tolerate small biases in the settings
distribution, we can explore how the protocol would have performed on
XOR 3 had we selected, prior to running the protocol, a nonzero
settings-bias parameter $\alpha$. We note that the protocol parameters
must be chosen prior to executing a secure protocol, and since we did
not choose a nonzero $\alpha$ in advance of examining XOR 3, we report
the following calculations only as a retrospective diagnostic. In
principle it is impossible to measure $\alpha$ through statistical
tests of the output of the random number generators that choose the
settings, because the settings probability can appear random,
unbiased, and independent even while changing from trial to trial
within the bounds of a potentially large $\alpha$. To
choose an example $\alpha$ to study, we examined 95 \% confidence
intervals for the individual settings probabilities from the six data
sets in \cite{shalm:2015}. The largest absolute difference from $0.5$
among the endpoints of these six intervals was $0.000211$ for Alice
and $0.000150$ for Bob. Assuming independence between Alice and Bob
(an assumption which was not contradicted by our statistical tests),
we computed the most and least likely measurement configurations given
this largest difference from $0.5$ for Alice's and Bob's settings
probability, and found that these would be contained in the interval
$(0.25-\alpha,0.25+\alpha)$ for $\alpha = 0.000181$. For this example
choice of $\alpha$, performing the modified optimization problem
described in \ref{st:T} yields a $T$ function with $m=0.01179$, and
for this $T$ function, the expected threshold computed according to
the 95 \% rule is $v_{\text{thresh}}=5.25 \times 10^5$, if we assume
the ``worst-case'' settings distribution among the six extremal
settings distributions that assign probability $0.25+\alpha$ to two
settings configurations and $0.25-\alpha$ to two other settings
configurations. This threshold is passed when the protocol is re-run
now with this non-zero $\alpha$. For $\epsilon_{\mathrm{p}}$ values of
$(0.01, 0.001,0.0001,0.00001)$ we get corresponding $-\log_2 \delta$
values of $(524, 383, 242, 101)$, which is a moderate reduction
compared to the corresponding values of $(582, 444, 306, 168)$
obtained by the running the protocol with $\alpha=0$. Alternatively,
we can fix $\epsilon_{\mathrm{p}}$ and study how $-\log_2\delta$
changes with $\alpha$. For one particular choice of
$\epsilon_{\mathrm{p}}=3.1797\times 10^{-4}$, which was the smallest
$\epsilon_{\mathrm p}$ value considered earlier in analyses of XOR 3,
$\alpha$ values of (0, 0.00001, 0.0001, 0.001) yield $-\log_2 \delta$
values of (367, 366, 321, 94). The largest value
  $\alpha=0.001$ in this list may be considered a conservative choice:
  if in the first calculation above we had used $99.999998\,\%$ instead of $95\,\%$ confidence
  intervals, we would have obtained a value of $\alpha\approx 0.001/3$
  instead of $\alpha=0.000181$.

\subsection{Performance of Previous Protocols.}
\label{st:previous}

Other protocols in the literature could not be used for our data sets
for various reasons. Many protocols
  apply to different measurement scenarios. For
  instance, \cite{colbeck:2011} describes a protocol involving three
  separated measurement stations, and while \cite{coudron:2014}
  provides impressive expansion rates and is secure against quantum
  side information, it requires eight separate devices. Other
  protocols exploring quantum side information in
Refs. \cite{miller:2014,chung:2014,miller2:2014} either also
apply to different experimental setups or provide only asymptotic
security results as the number of trials $n$ approaches
infinity. The first protocol achieving security against quantum
  side information \cite{vazirani:2012} applies to a bipartite
  experiment like ours but requires systems that achieve per-trial
  Bell violations much higher than ours. Another study
  \cite{thinh:2016} of bipartite experiments with data regimes
  characteristic of photonic systems applies to i.i.d.\ scenarios.

The protocols of Refs.~\cite{pironio:2010,arnon:2016} are applicable to our experimental scenario
  while making minimal assumptions, and given enough trials could work
  for any violation regime. Ref.~\cite{pironio:2010} obtained
  protocols for assumptions equivalent to ours, but considered
  also the case where the distributions are in addition assumed to
  be quantum achievable.
  Ref.~\cite{arnon:2016}, which uses the Entropy Accumulation Theorem of Ref.~\cite{dupuis:2016},
  obtained protocols assuming that the distributions are quantum
  achievable, but allowing for quantum side information. However, these
protocols are ineffective for the numbers of trials in our data sets,
which we illustrate with a heuristic argument. Both protocols are
based on the Clauser-Horne-Shimony-Holt (CHSH) Bell function
\cite{CHSH}
\begin{equation}
T^c(a,b,x,y) =\begin{cases}
1 & \text{ if } (x,y)\ne (1,1) \text{ and } a=b\\
1 & \text{ if } (x,y)= (1,1) \text{ and } a\ne b\\
0 & \text{ otherwise. }
\end{cases}
\end{equation}
The statistic $\overline {T^c}=n^{-1}\sum_{i=1}^nT^c_i$
used by these protocols for witnessing accumulated violation satisfies
$\mathbb{E}(\overline {T^c}) \le 0.75$ under LR, while $\mathbb{E}(\overline
{T^c})=0.75009787$ for the distribution in Table \ref{t:nosig}. The
completely predictable LR theory that only produces ``00'' outcomes
regardless of the settings satisfies $\mathbb{E}(\overline {T^c}) = 0.75$, but
in an experiment of $n=55,110,210$ trials, this theory can produce a
value of $\overline {T^c}$ exceeding 0.75009787 with probability
roughly $0.047$. Thus, based on this statistic alone, we cannot infer the
presence of any low-error randomness.

The protocol of Ref.~\cite{pironio:2010} (the PM protocol for short,
see \cite{pironio:2013,fehr:2013} for amendments), can be modified to
work with any Bell function, and there are methods for obtaining
better Bell functions \cite{nieto:2014,bancal:2014} or simultaneously
using a suite of Bell functions \cite{nieto:2016}. Here, we
demonstrate that for any choice of Bell function, the method of
\cite{pironio:2010} as refined in \cite{pironio:2013} cannot be
expected to effectively certify any randomness from an experiment
distributed according to Table \ref{t:nosig} unless the number of
trials exceeds $1.56\times 10^{8}$, which is larger than the
number of trials in our data runs.

For the most informative comparison to our protocol, we consider the
PM protocol without their additional constraint that the distribution
be induced by a quantum state. To derive a bound on the performance of
the PM protocol, we refer to Theorem 1 of \cite{pironio:2013}.  This
theorem involves a choice of Bell function denoted by $I$ (analogous
to our $T$), a threshold $J_{m}$ (analogous to our
$\Vthresh$) to be exceeded by the Bell estimator $\bar{I}
  = n^{-1}\sum_{i=1}^n I_i$, and a function $f$ that we discuss below.
To be able to extract some randomness, the theorem requires that
\begin{equation}\label{e:pirbound}
nf(J_m-\mu) > 0.
\end{equation}
The parameter $\mu$ is given by $(I_{\text{max}} +
I_{\text{NS}})\sqrt{(2/n)\ln(1/\epsilon)}$ where $I_{\text{max}}$ is
the largest value in the range of the Bell function $I$,
$I_{\text{NS}}\leq I_{\text{max}}$ is the largest possible expected
value of $I$ for non-signaling distributions, and $0<\epsilon\leq 1$
is a free parameter that is added to the TV distance from uniform for
the final output string. Smaller choices of $\epsilon$, which is
analogous to our $\epsilon_{\text{p}}$, are desirable but require
larger $n$ for the constraint \eqref{e:pirbound} to be positive as we
will see below. We also note that \eqref{e:pirbound} is a necessary
but not sufficient condition for extracting randomness; in particular,
we ignore the negative contribution from the parameter $\epsilon'$ of
\cite{pironio:2013} (somewhat analogous to the parameter $\kappa$ in the statement of the Protocol Soundness Theorem in \ref{st:pst}) as well as
any error introduced in the extraction step.

For \eqref{e:pirbound}, we can without loss of generality consider
only Bell functions for which $0 \le I_L < I_{\text{NS}}\le
I_{\text{max}}$, where $I_L$ is the maximum expectation of $I$ for LR
distributions. Further, because the relevant quantities below are
invariant when the Bell function is rescaled, we can assume $I_{L}=1$.
According to Ref.~\cite{pironio:2013}'s Eq.~8 and the following
paragraph, we can write
$f(x)=-\log_{2}(g(x))$, where $g$ is monotonically decreasing and
concave, and satisfies
\begin{equation}\label{e:maxpir}
\max_{ab}\mathbb{P}(ab|xy)\leq g(\mathbb{E}(I)_{\mathbb{P}})
\end{equation}
for all $xy$ and non-signaling distributions $\mathbb{P}$. (Recall that we are
not using the stronger constraint that $\mathbb{P}$ be induced by a quantum
state.)  According to \eqref{e:maxprobbound} we can define
$g(x)=1+(1-x)/(2(I_{\text{NS}}-1))$. Later we argue that
this definition of $g$ cannot be improved. Substituting into 
\eqref{e:pirbound} we get the inequality
\begin{equation}\label{e:pirminbound}
  - n \log_2\left[1+\frac{1-J_m  + (I_{\text{max}} + I_{\text{NS}})\sqrt{\frac{2}{n}\ln{\frac{1}{\epsilon}}}}{2(I_{\text{NS}}-1)}\right] >0.
\end{equation}
Since $2(I_{\text{NS}}-1)$ is
positive, this is equivalent to 
\begin{equation}\label{e:nbound1}
  \sqrt{\frac{2}{n}\ln{\frac{1}{\epsilon}}}<\frac{J_m-1}{I_{\text{max}}+I_{\text{NS}}}.
\end{equation}
Noting that $I_{\text{max}}+I_{\text{NS}}\ge 2I_{\text{NS}}$, this implies
\begin{equation}
  \sqrt{\frac{2}{n}\ln{\frac{1}{\epsilon}}} <\frac{J_m-1}{2I_{\text{NS}}}.
\end{equation}
Thus, the number of trials needed to extract randomness  by
the PM protocol is bounded below according to 
\begin{equation}
  \label{e:pirreq}
  n> 8\frac{\ln(1/\epsilon)I_{\text{NS}}^{2}}{(J_{m}-1)^{2}}.
\end{equation}
For a given anticipated experimental distribution $\mathbb{P}_{\text{ant}}$,
$J_m$ is best chosen to be at most $\mathbb{E}(I)_{\mathbb{P}_{\text{ant}}}$.  Otherwise,
the probability that $\bar I$ exceeds $J_{m}$ is small. However, for
the maximum amount of extractable randomness, $J_{m}$ should be close
to $\mathbb{E}(I)_{\mathbb{P}_{\text{ant}}}$.  Consider the inferred distribution from the first $5\times 10^6$ trials of Data Set 5. By following the procedure given in
Section 2 of \cite{bierhorst:2016}, we can write this distribution as
a convex combination of a PR box with weight $p=3.915\times 10^{-4}$
and an LR distribution with weight $1-p$. From this we see that one should choose $J_{m} \leq \mathbb{E}(I)_{\mathbb{P}_{\text{ant}}}
  = pI_{\text{NS}}+(1-p) \leq pI_{\text{NS}}+1$.  Substituting into
\eqref{e:pirreq} and using $\epsilon\le 0.05$ (a rather
high bound on the allowable TV distance from uniform) gives
\begin{equation}
  \label{e:prreq}
  n>8\frac{\ln(1/\epsilon)}{p^{2}}\geq 1.56\times 10^{8},
\end{equation}
which is already more than twice the number of trials used to generate randomness in Data Set 5. For smaller error values comparable to the ones we report, this bound only increases: achieving $\epsilon=10^{-12}$ would require at least $1.44 \times 10^{9}$ trials.

To finish our argument that the PM protocol cannot improve on this
bound under our assumptions, consider the definition of $g$.  If we
could find a function $g'\leq g$ with $g'(x)<g(x)$ for some
$x\in(1,I_{\text{NS}}]$, then $f=-\log_{2}(g')$ might yield a smaller
lower bound on $n$.  Note that for $x\leq 1$, $g'(x)\geq g'(1)$ and
$g'(1)$ must be at least 1 because, referring to \eqref{e:maxpir},
there is a conditionally deterministic LR distribution $\mathbb{P}$ satisfying
$\mathbb{E}(I)_{\mathbb{P}}=1$ and $\max_{ab}\mathbb{P}(ab|xy) =1$. Hence \eqref{e:pirbound} cannot
be satisfied for arguments $x$ of $f(x)=-\log_2(g'(x))$ with
$x\leq 1$. Given $x\in(1,I_{\text{NS}}]$, write
$x=(1-p)+pI_{\text{NS}}$. Let $\mathbb{Q}$ be the PR box achieving
$\mathbb{E}(I)_{\mathbb{Q}}=I_{\text{NS}}$ and $\mathbb{Q}'$ a conditionally deterministic LR
theory achieving $\mathbb{E}(I)_{\mathbb{Q}'}=1$.  Then
$\mathbb{E}(I)_{(1-p)\mathbb{Q}'+p\mathbb{Q}'}=x$. Furthermore, there is a setting $xy$ at which
the LR theory's outcome is inside the support of the PR box's
outcomes. To see this, by symmetry it suffices to consider the PR box
of \eqref{e:PRbox}. Its outcomes are opposite at setting $11$ and
identical at the other three. A deterministic LR theory's outcomes are
opposite at an even number of settings, so either it is opposite at
setting $11$, or it is identical at one of the others. For setting
$xy$, the bound in \eqref{e:maxpir} is achieved for our definition of
$g$. Hence any other valid replacement $g'$ for $g$ must satisfy
$g'(x)\ge g(x)$ for $x\in(1,I_{\text{NS}}]$, and so \eqref{e:pirbound}
with $f(x)=-\log_2(g'(x))$ implies \eqref{e:pirbound} with
$f(x)=-\log_2(g(x))$. Thus the lower bound on $n$ derived above will
apply to $g'$ as well.



\ignore{
\bibliographystyle{nature_nourl}
\bibliography{metabib}
}

\nolinenumbers
\begin{singlespace}

\end{singlespace}

\end{document}